\documentclass[a4paper]{article}
\usepackage[margin=1.25in]{geometry}
\usepackage{tikz}
\usetikzlibrary{positioning,chains,fit,shapes,calc}
\usepackage{ stmaryrd }
\usepackage{amsmath,amsthm,amssymb}
\usepackage{caption}
\usepackage{hyperref}

\newcommand{\strucA}{\mathcal A}
\newcommand{\strucB}{\mathcal B}
\newcommand{\strucC}{\mathcal C}

\newcommand{\graphG}{\mathcal G}
\newcommand{\graphH}{\mathcal H}
\newcommand{\graphK}{\mathcal K}
\newcommand{\graphC}{\mathcal C}

\newcommand{\classA}{\mathfrak A}
\newcommand{\classB}{\mathfrak B}
\newcommand{\classC}{\mathfrak C}
\newcommand{\classD}{\mathfrak D}
\newcommand{\classH}{\mathfrak H}

\newcommand{\Hom}{\ensuremath{\operatorname{Hom}}}

\newcommand{\classall}{\textbf{\_}}

\newcommand{\dset}{\operatorname{\mathsf D}}
\newcommand{\dsize}{\operatorname{\mathsf d}}
\newcommand{\dsub}[1]{\operatorname{\mathsf d}_{#1}}

\newcommand{\key}{\operatorname{key}}
\newcommand{\dom}{\operatorname{dom}}
\newcommand{\Trans}{\ensuremath{\operatorname{Trans}}}

\newcommand{\tw}{\operatorname{\mathsf tw}}
\newcommand{\ar}{\operatorname{\mathsf ar}}

\newtheorem{theorem}{Theorem}%
\newtheorem{lemma}[theorem]{Lemma}
\newtheorem{corollary}[theorem]{Corollary}
\newtheorem{claim}[theorem]{Claim}
\newtheorem{definition}[theorem]{Definition}
\newtheorem{example}[theorem]{Example}

\newcommand{\confORfull}[3]{#2} 

\title{A dichotomy for succinct representations of homomorphisms}

\date{September 29, 2022}

\author{
  Christoph Berkholz\\
  Technische Universit{\"a}t Ilmenau \\
  \texttt{christoph.berkholz@tu-ilmenau.de}
  \and
  Harry Vinall-Smeeth\\
  Humboldt-Universit{\"a}t zu Berlin \\
  \texttt{harry.vinall-smeeth@informatik.hu-berlin.de}
}

\begin{document}

\maketitle

\begin{abstract}
The task of computing homomorphisms between two finite relational
structures $\strucA$ and $\strucB$ is a well-studied question with
numerous applications. Since the set $\Hom(\strucA,\strucB)$ of all
homomorphisms may be very large having a method of representing it in
a succinct way, especially one which enables us to perform efficient
enumeration and counting, could be extremely useful.

One simple yet powerful way of doing so is to decompose
$\Hom(\strucA,\strucB)$ using union and Cartesian product. 
Such data structures, called d-representations, have been introduced by
Olteanu and Z{\'a}vodn{\`y} \cite{olteanu2015} in the context of
evaluating conjunctive queries. Their results also imply that if the treewidth of the left-hand
side structure $\strucA$ is bounded, then a d-representation of
polynomial size can be found in polynomial time. We show that for
structures of bounded arity this is optimal: if the treewidth is
unbounded then there are instances where the size of any
d-representation is superpolynomial. Along the way we develop tools
for proving lower bounds on the size of d-representations, in
particular we define a notion of reduction suitable for this context
and prove an almost tight lower bound on the size of d-representations
of all $k$-cliques in a graph.

\end{abstract}

\section{Introduction}

\begin{figure}  
  \centering

\begin{tikzpicture}
  [wire/.style={thick, ->, shorten <= 0.3mm, shorten >= 0.5mm},
   edge/.style={ -stealth, shorten <= 0.3mm, shorten >= 0.5mm},
   inputgate/.style={inner sep=1pt,minimum size=1mm},
   gate/.style={draw,circle,inner sep=1pt,minimum size=1mm},
   vertex/.style={draw,circle,fill=black,inner sep=2pt,minimum
   size=1.5mm}]

    \node at (0.5,-2) {$\graphG$};
    \node at (3.5,-2) {$\graphH$};
    
    \node at (9,-2) {d-representation of $\operatorname{Hom}(\graphG,\graphH)$};
 
    \node[vertex,label=left:{$x$}] (x1) at (0,0) {};
    \node[vertex,label=right:{$y$}] (x2) at (1,.577) {};
    \node[vertex,label=right:{$z$}] (x3) at (1,-.577) {};
    \draw[edge] (x1) -- (x2);
    \draw[edge] (x1) -- (x3);
    \draw[edge] (x2) -- (x3);

    \node[vertex,label=left:{$a_1$}] (a1) at (2.5,0+.4) {};
    \node[vertex,label=left:{$a_2$}] (a2) at (2.5,0) {};
    \node[vertex,label=left:{$a_3$}] (a3) at (2.5,0-.4) {};

    \node[vertex,label=right:{$b_1$}] (b1) at (2.5+1,1.2+.4) {};
    \node[vertex,label=right:{$b_2$}] (b2) at (2.5+1,1.2) {};
    \node[vertex,label=right:{$b_3$}] (b3) at (2.5+1,1.2-.4) {};

    \node[vertex,label=right:{$d_1$}] (d1) at (2.5+1,-1.2+.4) {};
    \node[vertex,label=right:{$d_2$}] (d2) at (2.5+1,-1.2) {};
    \node[vertex,label=right:{$d_3$}] (d3) at (2.5+1,-1.2-.4) {};

    \node[vertex,label=right:{$c$}] (c) at (2.5+2,0) {};

    \draw[edge] (a1) -- (b1);
    \draw[edge] (a2) -- (b1);
    \draw[edge] (a3) -- (b1);

    \draw[edge] (a1) -- (b2);
    \draw[edge] (a2) -- (b2);
    \draw[edge] (a3) -- (b2);

    \draw[edge] (a1) -- (b3);
    \draw[edge] (a2) -- (b3);
    \draw[edge] (a3) -- (b3);

    \draw[edge] (a1) -- (d1);
    \draw[edge] (a2) -- (d1);
    \draw[edge] (a3) -- (d1);

    \draw[edge] (a1) -- (d2);
    \draw[edge] (a2) -- (d2);
    \draw[edge] (a3) -- (d2);

    \draw[edge] (a1) -- (d3);
    \draw[edge] (a2) -- (d3);
    \draw[edge] (a3) -- (d3);

    \draw[edge] (b1) -- (c);
    \draw[edge] (b2) -- (c);
    \draw[edge] (b3) -- (c);

    \draw[edge] (c) -- (d1);
    \draw[edge] (c) -- (d2);
    \draw[edge] (c) -- (d3);

    \draw[edge] (a1) -- (c);
    \draw[edge] (a2) -- (c);
    \draw[edge] (a3) -- (c);

    \begin{scope}[xshift=9cm,yshift=-1.5cm]
      \node[inputgate] (xa1) at (-.8,0) {\tiny$x\!\mapsto\! a_1$};
      \node[inputgate] (xa2) at (0,0) {\tiny$x\!\mapsto\! a_2$};
      \node[inputgate] (xa3) at (+.8,0) {\tiny$x\!\mapsto\! a_3$};

      \node[inputgate] (yb1) at (-2.5-.8,0) {\tiny$y\!\mapsto\! b_1$};
      \node[inputgate] (yb2) at (-2.5 ,0) {\tiny$y\!\mapsto\! b_2$};
      \node[inputgate] (yb3) at (-2.5+.8,0) {\tiny$y\!\mapsto\! b_3$};

      \node[inputgate] (zd1) at (2.5-.8,0) {\tiny$z\!\mapsto\! d_1$};
      \node[inputgate] (zd2) at (2.5,0) {\tiny$z\!\mapsto\! d_2$};
      \node[inputgate] (zd3) at (2.5+.8,0) {\tiny$z\!\mapsto\! d_3$};

      \node[inputgate] (zc) at (-4,0) {\tiny$z\!\mapsto\! c$};
      \node[inputgate] (yc) at (4,0) {\tiny$y\!\mapsto\! c$};

      \node[gate] (g1) at (-1.5,1) {\small $\cup$};
      \node[gate] (g2) at (0,1) {\small $\cup$};
      \node[gate] (g3) at (1.5,1) {\small $\cup$};
      
      \node[gate] (g4) at (-1,2) {\small $\times$};
      \node[gate] (g5) at (1,2) {\small $\times$};
      
      \node[gate] (g6) at (0,2.75) {\small $\cup$};

      \draw[wire] (yb1) -- (g1);
      \draw[wire] (yb2) -- (g1);
      \draw[wire] (yb3) -- (g1);
      
      \draw[wire] (xa1) -- (g2);
      \draw[wire] (xa2) -- (g2);
      \draw[wire] (xa3) -- (g2);

      \draw[wire] (zd1) -- (g3);
      \draw[wire] (zd2) -- (g3);
      \draw[wire] (zd3) -- (g3);

      \draw[wire] (zc) -- (g4);
      \draw[wire] (yc) -- (g5);

      \draw[wire] (g1) -- (g4);
      \draw[wire] (g2) -- (g4);

      \draw[wire] (g3) -- (g5);
      \draw[wire] (g2) -- (g5);

      \draw[wire] (g4) -- (g6);
      \draw[wire] (g5) -- (g6);

    \end{scope}
    
\end{tikzpicture}

 \caption{A deterministic d-representation of all homomorphisms from
  $\graphG$ to $\graphH$.}
\label{fig:representation}
\end{figure}

The task of computing homomorphisms between two finite relational
structures has a long history and numerous applications. Most notably,
as pointed out by Feder and Vardi \cite{feder1993},
it is the right abstraction for the constraint satisfaction problem
(CSP)---a framework for search problems that generalised
 Boolean satisfiability. Moreover, evaluating conjunctive queries on a
relational database is equivalent to computing homomorphisms from the
query structure to the database.
While deciding the existence of a homomorphism from a structure
$\strucA$ to a structure $\strucB$ 
is a classical NP-complete problem, several restrictions of the input
instance have been considered in order to understand the landscape of
tractability. One line of research investigates \emph{right-hand-side}
restrictions, where it is asked for which classes of structures $\strucB$ the
CSP becomes tractable and when it remains hard. This culminated in the solution
\cite{DBLP:conf/focs/Bulatov17,zhuk2020proof} of the CSP-dichotomy
 conjecture  \cite{feder1993} that
 characterises those $\strucB$ where finding a homomorphism from a
 given structure $\strucA$ can be done in polynomial time (assuming
 P$\neq$NP).

 Another line of research, to which we contribute in this paper, focuses on
 \emph{left-hand-side} restrictions: for which classes of structures $\strucA$ can
 we efficiently find a homomorphism from $\strucA$ to a given 
 $\strucB$?
 In this scenario, a dichotomy is only known when all relations have
 bounded arity, as is the case for graphs, digraphs, or $k$-uniform
 hypergraphs. Grohe \cite{Grohe2007} showed that, modulo complexity
 theoretic assumptions, for any class of structures $\classA$
 of bounded arity the decision
 problem, ``Given a structure $\strucA \in \classA$ and a
 structure $\strucB$, is there a homomorphism from $\strucA$ to
 $\strucB$?'' is in polynomial time if and only if the homomorphic
 core of every structure in $\classA$ has bounded treewidth.
 For classes of unbounded arity, polynomial time tractability has been
 shown for fractional hypertreewidth
 \cite{Atserias:2013,grohe2014constraint}, but a full characterisation
 of tractability has
 only been obtained in the parameterised setting using submodular
 width \cite{marx2013}.
Besides deciding the existence of a homomorphism, the complexity of
counting all homomorphism has also been characterised in the
right-hand-side regime \cite{DBLP:journals/jacm/Bulatov13} and for bounded-arity classes of left-hand-side
structures \cite{DBLP:journals/tcs/DalmauJ04}.
A third task, that is less well understood, is to \emph{enumerate} all
homomorphisms; here
only partial results on the complexity are known (e.\,g.\ \cite{bulatov_et_al:LIPIcs:2009:1838,DBLP:journals/ita/CreignouH97,DBLP:conf/stacs/SchnoorS07,DBLP:journals/constraints/GrecoS13}).

In this work we consider the task of \emph{representing} the set $\operatorname{Hom}(\strucA,\strucB)$ of all
homomorphisms in a succinct and accessible way. In particular, we want
to store all, potentially exponentially many, homomorphisms, in a data
structure of polynomial size that enables us to, e.\,g.\, generate a
stream of all homomorphisms. The data structures we are interested in---so-called d-representations---
were first introduced to represent homomorphisms in the context of join evaluation under the name
\emph{factorised databases} \cite{olteanu2015}.
They are conceptually very simple: the set of homomorphisms is
represented by a circuit, where the ``inputs'' are mappings of single
vertices and larger sets of mappings are generated by combining local
mappings using Cartesian product $\times$ and union
$\cup$. In the circuit previously computed sets of local
homomorphisms are represented by gates and can be used several times,
see Figure~\ref{fig:representation} for an example.
Such a representation is called \emph{deterministic} if every $\cup$-gate is
guaranteed to combine disjoint sets. Deterministic representations
have the advantage that the number of homomorphisms can be efficiently
counted by adding the sizes of the local homomorphism sets on every
$\cup$-gate and multiplying them on every $\times$-gate. Moreover, all
homomorphisms can be
efficiently enumerated where the delay between two outputs is only
linear in the size of every produced homomorphism (= size of the
universe of $\mathcal A$) \cite{DBLP:conf/icalp/AmarilliBJM17}.
It is known that if the treewidth of the left-hand side
structure is bounded, then a deterministic d-representation of
polynomial size can be
found in polynomial time \cite{olteanu2015}. Our main theorem shows that for
structures of bounded arity this is optimal: if the treewidth is
unbounded, then there are instances where the size of any (not necessarily deterministic)
representation is superpolynomial.

\begin{theorem}\label{thm:maintheorem}
  Let $r\in\mathbb N$, $\sigma$ a signature of arity $\leq r$ and $\classA$ a
  class of $\sigma$-structures. Then the following are equivalent:
  \begin{enumerate}
  \item There is a $w\in\mathbb N$ such that every structure in
    $\classA$ has treewidth at most $w$.
  \item A deterministic
    d-representation of $\operatorname{Hom}(\strucA,\strucB)$ can be
    computed in polynomial time, for any $\strucA \in \classA$ and any $\strucB$. 
  \item There is a $c\in\mathbb N$ such that for any $\strucA\in\classA$
    and any $\strucB$ there exists a (not necessarily deterministic) d-representation of
    $\operatorname{Hom}(\strucA,\strucB)$ of size $O\big(\,(\|\strucA\|+\|\strucB\|)^c\,\big)$.
  \end{enumerate}
\end{theorem}

\subparagraph*{Related work}
The research on succinct data structures for homomorphism problems has
emerged from the two different perspectives. When fixing the
right-hand-side structure $\strucB$, then data structures
like multi-valued decision
diagrams (MDD) \cite{AmilhastreCompilingCSPsComplexity2014}, AND/OR multi-valued
decision diagrams (AOMDD) \cite{MateescuCompilingConstraintNetworks2006}, and
multi-valued decomposable decision graphs (MDDG)
\cite{KoricheCompilingConstraintNetworks2015} have been proposed,
which arose from representations for Boolean functions that are
studied in \emph{knowledge compilation} (see, e.\,g.,
\cite{DBLP:journals/jair/DarwicheM02}). The (deterministic)
d-representations studied in this paper can be interpreted as (deterministic) DNNF circuits
with zero-suppressed semantics \cite[Lemma~7.4]{DBLP:conf/icalp/AmarilliBJM17}, where a $\cup$-gate corresponds to a
(deterministic) $\lor$-gate and a $\times$-gate corresponds to a
\emph{decomposable} $\land$-gate.

In the left-hand-side regime, representations have been introduced in
the context of enumerating query results. Most notably,
Olteanu and Z{\'a}vodn{\`y} \cite{olteanu2015} introduced the notion
of \emph{factorised databases} that are used to decompose the result
relation of a conjunctive query using Cartesian product and
union. Their findings imply the upper bound part of our dichotomy
theorem: if $\strucA$ has bounded treewidth, its tree decomposition
defines a so-called \emph{d-tree}, which structures the polynomial
size d-representation. 
They have also shown a limited lower bound for \emph{structured} representations
(``d-representations respecting a d-tree''). However, this lower bound
considers only a small subclass of all possible d-representations\confORfull{.}{, see
Appendix~\ref{sec: extendedUB} for a discussion.}{, see
Appendix~\ref{sec: extendedUB} for a discussion.}  %
In a similar vein, in knowledge compilation there exist several
restrictions of DNNFs e.\,g.\ requiring $\vee$-gates to be 
decision or deterministic \cite{DBLP:journals/jair/DarwicheM02}, or
enforcing structuredness \cite{DBLP:conf/aaai/PipatsrisawatD08}. In this light, the significance of our lower
bound comes from the fact that it holds for the most general
notion of representations (d-representations), which correspond to unrestricted DNNFs.

The proof of our lower bound has some connections to the
conditional lower bound for the counting complexity of homomorphisms
\cite{DBLP:journals/tcs/DalmauJ04}, which in turn builds upon the
construction of Grohe \cite{Grohe2007}. The essence of these proofs
is to rely on an assumption about the hardness of the parametrised clique problem and
reduce this to all structures of unbounded treewidth. We take a
similar route: in Section~\ref{s: clique} we prove an unconditional
lower bound for representing cliques and obtain our main lower bound
using a sequence of reductions in Section~\ref{s: dichotomy}.

The circuit notion for representing the set of homomorphisms between
two given structures (or, equivalently, the result relation of a multiway
join query) in a succinct data structure might be confused with previous
work on the (Boolean or arithmetic) circuit complexity for deciding or
counting homomorphisms or subgraphs. In this research branch, a structure $\mathcal B$  over a universe of size
$n$ is given as input to a circuit $C_{\mathcal A, n}$, which decides the existence of or counts the
number of homomorphisms (or subgraph-embeddings) from $\mathcal A$ to $\mathcal B$. 
Examples include monotone circuits for finding cliques
\cite{DBLP:journals/combinatorica/AlonB87,DBLP:journals/siamcomp/Rossman14}, bounded-depth circuits for finding cliques and other
small subgraphs \cite{DBLP:conf/stoc/Rossman08,DBLP:journals/siamcomp/LiRR17} as well as graph polynomials and monotone arithmetic
circuits \cite{DBLP:journals/jgaa/Engels16,DBLP:conf/fsttcs/BlaserKS18,DBLP:conf/icalp/Komarath0R22} for counting
homomorphisms. In particular, the recent work of Komarath, Pandey and Rahul \cite{DBLP:conf/icalp/Komarath0R22}
studies monotone arithmetic circuits that have, for each pattern $\graphG$ and each $n$, an
input indicator variable $x_{\{u,v\}}$ for each potential edge $\{u,v\} \in
[n]^2$ in the second graph $\graphH$. For every input (i.e. setting indicator
variables according to a graph $\graphH$ on $n$
vertices), the arithmetic circuit has to compute the number of homomorphisms from
$\graphG$ to $\graphH$. 
Interestingly, Komorath et al. prove a tight bound and show that such arithmetic circuits need
size $n^{\operatorname{tw}(\graphG)+1}$. Unfortunately, this and related
results from circuit complexity (such as lower bounds for the clique problem) do not translate to the knowledge
compilation approach. Part of the reason is that we crucially have a different representation
\emph{for each pair $\graphG$, $\graphH$} and having, e.\,g.\, an arithmetic circuit
computing the constant number $|\!\operatorname{Hom}(\graphG,\graphH)|$ is
trivial. Moreover, due to monotonicity, the worst-case right-hand-side
instances $\graphH$ in \cite{DBLP:conf/icalp/Komarath0R22} are
complete graphs, whereas d-representations lack this form of
monotonicity: adding edges to $\graphH$ can make factorisation simpler
and in particular occurrences of patterns in complete graphs can be succinctly factorised.

Despite this, some techniques on a more general
level (e.\,g.\ arguing about the transversal of a circuit or using random
graphs as bad examples) are useful in circuit complexity as well as
for proving lower bounds on representations.

\section{Preliminaries}
We write $\mathbb{N}$ for the set of non-negative integers and define
$[n]:= \{1, \dots, n\}$ for any positive integer $n$. Given a set $S$
we write $2^S$ to denote the power set of $S$. Whenever writing $a$ to denote
a $k$-tuple, we write $a_i$ to denote the tuple's $i$-th component;
i.\,e., $a=(a_1,\dots, a_k)$. For a function $f \colon X \to Y$ and $X' \subset X$ we write $\pi_{X'} f$ to denote the projection of $f$ to $X'$. Given a set of functions, each of which has a domain containing $X'$,  we write $\pi_{X'}F := \{ \pi_{X'}f \, \mid \, f \in F \}$. 

\subparagraph*{Graphs, Minors, Structures, Tree Decompositions}
Whenever $\graphG$ is a graph
or a hypergraph we write $V(\graphG)$ and $E(\graphG)$ for the set of nodes and
the set of edges, respectively, of $\graphG$. We let $\graphK_k$ be
the complete graph on $k$ vertices, $\graphC_k$ the $k$-cycle graph, and
$\graphG_k$ the $k\times k$-grid graph. Given a graph $\graphG$ and
$\{u,v\} \in E(\graphG)$, we can form a new graph by \emph{edge contraction}: replacing $u$ and $v$ be a new vertex $w$ adjacent to all neighbours of $u$ and $v$. A graph $\graphH$ is a \emph{minor} of $\graphG$ if $\graphH$ can be obtained from $\graphG$ by repeatedly deleting vertices, deleting edges and contracting edges.

A \emph{tree decomposition} of a graph $\graphG$ is a pair
$(T,\beta)$ where $T$ is a tree and $\beta\colon V(T)\to 2^{V(\graphG)}$ associates to every node $t\in V(T)$ a bag $\beta(t)$ such
that the following is satisfied:
(1) For every $v\in V(\graphG)$ the set $\{t\in V(T)\mid v\in
  \beta(t)\}$ is non-empty and forms a connected set in $T$.
(2) For every $\{u,v\}\in E(\graphG)$ there is some $t\in V(T)$ such that $\{u,v\}\subseteq
  \beta(t)$.
The width of a tree decomposition is $\max_{t\in V(T)}|\beta(t)|-1$
and the \emph{treewidth of} $\graphG$ is the minimum width of any
tree decomposition of $\graphG$.

A (relational) signature $\sigma$ is a set of relation symbols $R$, each of which is equipped with an arity $r =r(R)$. A (finite, relational) $\sigma$-structure $\strucA$
consists of a finite universe $A$ and  relations
$R^{\strucA}\subseteq A^r$ for every $r$-ary relation
symbol $R\in \sigma$. We will write $\|\strucA\|:=
\sum_{R \in \sigma} |R^{\strucA}|$. The \emph{Gaifman graph} of
$\strucA$ is the graph with vertex set $A$ and edges $\{u,v\}$ for any
distinct $u,v$ that occur together in a tuple of a relation in
$\strucA$. The \emph{treewidth} of a structure is the treewidth of its
Gaifman graph. We say $\strucA$ is \emph{connected} if its Gaifman graph is connected and we will henceforth assume, without loss of generality, that all structures in this paper are connected.

\subparagraph*{Enumeration}
An enumeration algorithm for $\Hom(\strucA,\strucB)$ proceeds in two
stages. In the \emph{preprocessing} stage the algorithm does some
preprocessing on $\strucA$ and $\strucB$. In the \emph{enumeration phase} the algorithm enumerates, without repetition, all homomorphisms in $\Hom(\strucA,\strucB)$, followed by the end of enumeration message. The \emph{delay} is the maximum of three times: the time between the start of the enumeration phase and the first output homomorphism, the maximum time between the output of two consecutive homomorphisms and between the last tuple and the end of enumeration message. The \emph{preprocessing time} is the time the algorithm spends in the preprocessing stage, which may be $0$. Similarly given a d-representation $C$ for $\Hom(\strucA,\strucB)$, an enumeration algorithm for $C$ has a preprocessing stage, where it can do some preprocessing on $C$, and an enumeration phase defined as above.
  
\section{Homomorphisms and the complexity of constraint satisfaction}

 A \emph{homomorphism} $h\colon \strucA\to
\mathcal B$ between two $\sigma$-structures $\strucA$ and $\mathcal
B$ is a mapping from $A$ to $B$ that preserves
all relations, i.\,e., for every $r$-ary $R\in\sigma$ and $(a_1,\ldots,a_r)\in A^r$ it holds that
if $(a_1,\ldots,a_r)\in R^{\strucA}$, then $(h(a_1),\ldots,h(a_r))\in R^{\mathcal B}$.
We let $\Hom(\strucA, \mathcal B)$ be the set of all homomorphisms
from $\strucA$ to $\mathcal B$. A (homomorphic) core of a
structure $\strucA$ is an inclusion-wise minimal substructure $\strucA'\subseteq
\strucA$ such that there is a homomorphism from $\strucA$ to
$\strucA'$.
It is well known that all cores of a structure are
isomorphic, hence we will also speak of \emph{the} core of a structure. 

Following common notation we fix a (potentially infinite) signature $\sigma$ and define for classes of $\sigma$-structures $\classA$ and $\classB$ the (promise) decision problem CSP($\classA$, $\classB$) to be:
``Given two $\sigma$-structures $\strucA \in \classA$ and
  $\mathcal B\in \classB$, is there a homomorphism from $\strucA$ to
  $\mathcal B$?''
Similarly, the counting problem \#CSP($\classA$, $\classB$) asks:
  ``Given two $\sigma$-structures $\strucA \in \classA$ and
  $\mathcal B\in \classB$, what is the number of homomorphisms from $\strucA$ to $\mathcal B$?''
A lot of work has been devoted towards classifying the classes
of structures for which the problems are solvable in polynomial time. Normally either the left-hand-side $\classA$ or
the right-hand-side $\classB$ is restricted and the other part
($\classB$ or $\classA$) is the class $\classall$ of all structures. 
A related problem is Enum-CSP($\classA$, $\classB$) \cite{bulatov_et_al:LIPIcs:2009:1838}, which is the following task:
  ``Given two $\sigma$-structures $\strucA \in \classA$ and
  $\mathcal B\in \classB$, enumerate all homomorphisms from $\strucA$ to $\mathcal B$''.
One way of defining tractability for enumeration algorithms is 
\emph{polynomial delay enumeration}, where the preprocessing time and
the delay is polynomial in $\strucA$ and $\strucB$.

In this paper we focus on ``left-hand-side'' restrictions, where
$\classB$ is the class of all structures. Moreover, we assume that
the arity of each symbol in $\sigma$ is bounded by some constant
$r$. In this setting the complexity of CSP($\classA$, $\classall$) and
\#CSP($\classA$, $\classall$) is fairly well understood: the
decision problem is polynomial time tractable iff the core of every
structure in $\classA$ has bounded treewidth, while the counting
problem is tractable if every structure from $\classA$ itself has
bounded treewidth.
This is made precise by the following two theorems.

\begin{theorem}[\cite{Grohe2007}]\label{thm:decision-dichotomy}
  Let $r\in\mathbb N$, $\sigma$ be a signature of arity $\leq r$ and $\classA$ a
  class of $\sigma$-structures.
  Under the assumption that there is no $c\in\mathbb N$ such that for
  every $k\in\mathbb N$
  there is an algorithm that finds a $k$-clique in an $n$-vertex graph in time $O(n^c)$ the
  following two statements are equivalent.
  \begin{enumerate}
  \item There is a $w\in\mathbb N$ such that the core of
    every structure in $\classA$ has treewidth at most $w$.
  \item {\upshape CSP($\classA$,$\classall$)} is solvable in polynomial time.
  \end{enumerate}
\end{theorem}

\begin{theorem}[\cite{DBLP:journals/tcs/DalmauJ04}]\label{thm:counting-dichotomy}
  Let $r\in\mathbb N$, $\sigma$ be a signature of arity $\leq r$ and $\classA$ a
  class of $\sigma$-structures.
  Under the assumption that there is no $c\in\mathbb N$ such that for
  every $k\in\mathbb N$
  there is an algorithm that counts the number of $k$-cliques in an $n$-vertex graph in time $O(n^c)$ the
  following two statements are equivalent.
  \begin{enumerate}
  \item There is a $w\in\mathbb N$ such that every structure in $\classA$ has treewidth at most $w$.
  \item {\upshape \#CSP($\classA$,$\classall$)} is solvable in polynomial time.
  \end{enumerate}
\end{theorem}

To understand the difference between these characterisations, consider the class $\classA$ of all
structures $\strucA_k$ that are complete graphs on $k$ vertices with an additional vertex with a self-loop. The
homomorphic core of such structures is just the self-loop and finding one homomorphism
from $\strucA_k$ to $\strucB$ is equivalent to finding a self-loop in
$\strucB$. However, counting homomorphisms from $\strucA_k$ to $\strucB$ is as hard as counting
$k$-cliques: if $\strucB$ is a simple graph $\graphG$ with one additional vertex
with a self-loop, then the number of homomorphisms from $\strucA_k$ to
$\strucB$ is the number of $k$-cliques in $\graphG$ plus one.

The complexity of the corresponding enumeration problem
Enum-CSP($\classA$, $\classall$) is still open. It has been shown
that polynomial delay enumeration is possible if $\classA$ has
bounded treewidth \cite{bulatov_et_al:LIPIcs:2009:1838}. On the
other hand, polynomial delay enumeration implies solvability of the
decision problem in polynomial time (because either the first solution
or an end-of-enumeration message has to
appear after polynomial time). Hence it follows from Theorem~\ref{thm:decision-dichotomy}, under
the same complexity assumption, that there is no polynomial delay
enumeration algorithm if the cores of the structures in $\classA$ have unbounded treewidth. For further discussion on this topic we refer the reader to \cite{bulatov_et_al:LIPIcs:2009:1838}.

Our main result (Theorem~\ref{thm:maintheorem}) can be viewed as an 
unconditional dichotomy for enumeration and counting in a restricted
class of algorithms: when the algorithm relies on local
decompositions into union and product, then the tractable instances
are exactly those that have bounded treewidth. Interestingly, this
matches the conditional dichotomy for the counting case (Theorem~\ref{thm:counting-dichotomy}).

\section{Factorised Representations} \label{s: FR}

In this section we formally introduce the factorisation formats for
CSPs. These formats agree with the factorised representations of
relations introduced by Olteanu and Z{\'a}vodn{\'y} \cite{olteanu2015} in the context of evaluating
conjunctive queries on relational databases. While we stick to the
naming conventions introduced there we provide a
slightly different circuit-based definition that is very
much inspired by \cite{DBLP:conf/icalp/AmarilliBJM17} and the notion
of \emph{set circuits} introduced in
\cite{DBLP:conf/pods/AmarilliBMN19}. 

A \textit{factorisation circuit} $C$ for two sets $A$ and $B$ is an acyclic directed graph with node
labels and a unique sink. Each node without incoming edges is called
an \textit{input gate} and
labelled by $\{a\mapsto b\}$ for some $a\in A$ and $b\in B$. Every
other node is labelled by either $\cup$ or $\times$ and called a
\textit{$\cup$-gate} or \textit{$\times$-gate}, respectively.
For each gate $g$ in the circuit we inductively define its \textit{domain}
$\dom(g)\subseteq A$ by $\dom(g)=\{a\}$ if $g$ is an input gate with
label $\{a\mapsto b\}$ and $\dom(g)=\bigcup^r_{i=1}\dom(g_i)$ if $g$ is a
non-input gate with child gates $g_1,\ldots,g_r$.

A factorisation circuit is \textit{well-defined} if for every gate $g$
with child gates $g_1,\ldots,g_r$ it holds that
$\dom(g)=\dom(g_1)=\cdots=\dom(g_r)$ if $g$ is a $\cup$-gate and
$\dom(g_i)\cap\dom(g_j)=\emptyset$ for all $i\neq j$ if $g$ is a
$\times$-gate.
For every gate $g$ in a well-defined factorisation circuit we let $S_g$ be a set
of mappings $h\colon \dom(g)\to B$ defined by
\begin{equation}
  \label{eq:1}
  S_g :=
  \begin{cases}
    \big\{\{a\mapsto b\}\big\}&\text{if $g$ is an input labelled by
      $\{a\mapsto b\}$}\\
    S_{g_1} \cup \cdots \cup S_{g_r}&\text{if $g$ is a $\cup$-gate with children $g_1,\ldots,g_r$,}\\
    \big\{\,h_1\cup\cdots \cup h_r\,\mid\,h_i\in S_{g_i},\, i\in [r]\,\big\}&\text{if $g$ is a $\times$-gate with children $g_1,\ldots,g_r$.}\\
  \end{cases}
\end{equation}

We define $S_C := S_s$ for the sink $s$ of $C$. For each
gate $g$ we let $C_g$ be the sub-circuit with sink $g$. By
$\|C\|$ we denote the size of a factorisation
circuit $C$, which is defined to be the number of gates plus the
number of wires. The number of gates in $C$ is denoted by $|C|$.

Before defining factorised representations for CSP-instances, we
introduce two special types of circuits. A
factorisation circuit is \textit{treelike} if the underlying graph is
a tree, i.\,e., every non-sink gate has exactly one parent.
Moreover, a well-defined factorisation circuit is \textit{deterministic} if for
every $\cup$-gate $g$ the set $S_g$ is a \textit{disjoint} union of its child
sets $S_{g_1},\ldots,S_{g_r}$. Note that while treelikeness is a
syntactic property of the circuit structure, being deterministic is a
semantic property that depends on the valuations of the gates.
Now we are ready to state a circuit-based definition of the
factorised representations defined in \cite{olteanu2015}.

\begin{definition}
Let $\mathcal A$ and $\mathcal B$ be two $\sigma$-structures.
\begin{enumerate}
\item 
A (deterministic) \emph{d-representation} for $\mathcal A$ and $\mathcal B$
is a well-defined (deterministic) factorisation circuit over
$V(\mathcal A)$ and $V(\mathcal B)$ where $S_C=\operatorname{Hom(\mathcal A, \mathcal
  B)}$.
\item A (deterministic) \emph{f-representation} is a (deterministic) d-representation with the additional
restriction that the circuit is treelike.  
\end{enumerate}
\end{definition}

For brevity we will sometimes refer to d/f-representations as d/f-reps. Note that a d-rep can be more succinct than a
f-rep and we will mostly deal with d-reps in
this paper. However, in the proofs it will sometimes be convenient to
expand out the circuit in order to make it treelike. More formally,
the \emph{transversal} $\Trans(C)$ of a d-rep $C$ is the
f-rep obtained from $C$ as follows: using a top-down
transversal starting at the output gate, we replace each gate $g$ with
parents $p_1,\ldots,p_d$ by $d$ copies $g_1,\dots,g_d$ such that the
in-edges of each $g_i$ are exactly the children of $g$ and $g_i$ has
exactly one out-edge going to $p_i$. This procedure produces a
treelike circuit that is well-defined/deterministic if $C$ was
well-defined/deterministic. Finally it can easily be verified
that $S_{\Trans(C)} = S_C$.

We will often want to construct new factorised circuits from old
ones. The following  lemma introduces two constructions that will be
particularly useful, the proof of correctness can be found
\confORfull{in the full version of this paper.}{in Appendix~\ref{ap:
    technical}.}{in Appendix~\ref{ap: technical}.} 

\begin{lemma} \label{lem:technical}
Let $\mathcal{A}, \mathcal{B}$ be $\sigma$-structures and $C$ be a d-rep of $\Hom(\mathcal{A}, \mathcal{B})$. Let $X=\{x_1,\dots,x_{\ell}\} \subseteq A$, $Y_1, \dots, Y_{\ell} \subseteq B$, $\ell \ge 1$. Then one can construct the following factorised circuits in time $O(\|C\|)$.
\begin{enumerate}
\item $C'$, such that $S_{C'}= \pi_{X}\Hom(\mathcal{A},\mathcal{B})$ and $\|C'\| \le \|C\|$.
\item $C''$, such that $S_{C''} = \{ h \in \Hom(\mathcal{A},\mathcal{B}) \, \mid \, h(x_i) \in Y_i, i \in [\ell] \}$ and $\|C''\| \le \|C\|$. 
\end{enumerate}
\end{lemma}

A special f-rep is the \emph{flat representation}:
a depth-2 circuit with a single $\cup$-gate at the top followed by a
layer of $\times$-gates. Note that for any pair $\mathcal A$,
$\mathcal B$, of $\sigma$-structures the flat representation has size $1+|\!\Hom(\mathcal A,\mathcal B)|\cdot (2|A|+2)$. Intuitively, this representation corresponds to listing all homomorphisms and provides a trivial upper bound on representation size.

Deterministic d-representations have two desirable properties:  
they allow us to compute $|\!\Hom(\strucA,\strucB)|$ in time $O(\|C\|)$ and
to enumerate all homomorphisms with $O(|A|)$
delay after $O(\|C\|)$ preprocessing. Efficient counting is possible by computing bottom-up the
number $|S_g|$ for each gate using multiplication on every
$\times$-gate and summation on every (deterministic) $\cup$-gate. If, additionally, $C$ is normal---i.\,e., no parent
  of a $\cup$-gate is a $\cup$-gate and the in-degree of every
  $\cup$- and $\times$-gate is at least 2---Olteanu and Z{\'a}vodn{\'y} \cite[Theorem 4.11]{olteanu2015} show 
enumeration with $O(|A|)$ delay and \emph{no} preprocessing is possible by sequentially enumerating the sets $S_{g_i}$ of every child of a (deterministic) $\cup$-gate and by a nested loop to generate all combinations of child elements at $\times$-gates. The case where $C$ is not normal is shown in \cite[Theorem~7.5]{DBLP:conf/icalp/AmarilliBJM17} and is more involved. Note
that the delay is optimal in the sense that every homomorphism that is
enumerated is of size $O(|A|)$.

In the other direction this means that constructing a deterministic
d-rep
is at least as hard as counting the number of homomorphisms. Our main
theorem implies that, modulo the same assumptions as Theorem~\ref{thm:counting-dichotomy},
 the opposite is also true: for a class $\classA$ of
structures of bounded arity there is a polynomial time algorithm that constructs a
d-representation of polynomial size for two given structures $\mathcal
A\in\classA$ and $\mathcal B$ if and only if, there is a
polynomial-time algorithm that counts the number of homomorphisms
between $\mathcal
A\in\classA$ and $\mathcal B$.

\subparagraph*{Upper bounds on representation size} %
We have already argued that there is always a flat representation of
size $O(|A|\cdot |\!\Hom(\mathcal A, \mathcal B)|)$. Thus, as a
corollary of \cite{Atserias:2013} we get an upper bound of $O(|A|\cdot \|\mathcal
B\|^{\rho^\ast(\mathcal A)})$, where $\rho^\ast(\mathcal A)$ is the
\emph{fractional edge cover number of $\strucA$}. Note that, however, the fractional edge
cover number for structures of
bounded arity is quite large. More precisely, if all relations in
$\mathcal A$ have arity at most $r$, then $\rho^\ast(\mathcal A) \geq
\frac1r|A|$. 

Luckily in many cases we can do better: the results by Olteanu and
Z{\'a}vodn{\'y} in \cite{olteanu2015} imply that given a
tree-decomposition of $\mathcal{A}$ of width $w-1$ we can construct a
d-rep of $\Hom(\mathcal{A},\mathcal{B})$ of size
$O(\|\mathcal{A}\|^2 \|\mathcal{B}\|^{w})$ in time
$O(\operatorname{poly}(\|\mathcal{A}\|) \|\mathcal{B}\|^{w}
\log(\|\mathcal{B}\|))$. Moreover the d-reps produced are normal and
deterministic, meaning they allow us to perform efficient enumeration
and counting. Therefore if $\classA$ is a class of bounded treewidth 
this gives us one method for solving \#CSP($\classA$,$\classall$) in
polynomial time. In fact, the same holds true if $w$ is the more general
\emph{fractional hypertreewidth}, although for the case of bounded
arity structures the two measure differ only by a constant. We
discuss the unbounded arity case in the conclusion \confORfull{and, in more detail, in the full version of this paper.}{and Appendix~\ref{sec: extendedUB}.}{and Appendix~\ref{sec: extendedUB}.}

\section{A near-optimal bound for cliques} \label{s: clique}

The goal of this section is to prove the following two theorems:

\begin{theorem} \label{thm: f-clique}

For any $k \in \mathbb{N}$ there exist arbitrary large graphs $\graphG$
with $m$ edges
such that any f-rep of $\Hom(\graphK_k,\graphG)$ has size
$ \Omega ( m^{k/2} / \log^{k}(m) ) $. 
  
\end{theorem}

\begin{theorem} \label{thm: d-clique}

  For any $k \in \mathbb{N}$ there exist arbitrary large graphs $\graphG$
  with $m$ edges
  such that any d-rep of $\Hom(\graphK_k,\graphG)$ has size
   $ \Omega ( m^{k/2} / \log^{3k-1}(m) ) $.

\end{theorem}

These bounds are almost tight since the number of
$k$-cliques in a graph with $m$ edges is bounded by
$m^{k/2}$. Moreover Theorem~\ref{thm: d-clique} is a crucial
ingredient for proving our main theorem in Section~\ref{s: dichotomy}.  We
will first prove Theorem~\ref{thm: f-clique} and then show how this implies the bound for d-reps.

The main idea is to exploit a correspondence between the structure of a (simple) graph $\graphG$ and
f-reps of $\Hom(\graphK_k,\graphG)$. To illustrate this
consider the case $k=2$, where $V(\graphK_2) =\{x_1, x_2\}$ and each
 $h \in \Hom(\graphK_2,\graphG)$ corresponds to an edge of $\graphG$. Let $C$ be a  f-rep
of $\Hom(\graphK_2,\graphG)$, with $\times$-gates $g_1, \dots,
g_{\alpha}$. Each
$g_i$ has two children $g^1_i$, $g^2_i$ with $\dom(g^1_i) = x_1$ and
$\dom(g^2_i)=x_2$. Since no $\times$-gates can occur in $C_{g^1_i}$
or $C_{g^2_i}$, $S^1_{g_i} =  \{  \{x_1 \mapsto a\}
\mid a \in A_i \}$
and $S_{g^2_i} =  \{  \{x_2 \mapsto b\} \mid b \in B_i \}$ for some disjoint $A_i,
B_i \subseteq V(\graphG)$. Therefore $A_i \times B_i$ is a complete bipartite
subgraph of $\graphG$. Since the ancestors of each $\times$-gate can only be
$\cup$-gates, each f-rep of $\Hom(\graphK_2,\graphG)$ corresponds to a set of complete bipartite subgraphs that cover every
edge of $\graphG$. Finding such sets and investigating their properties has
been studied in various contexts, for example see \cite{Chung1983,
  fleischner2009, jukna2009, mubayi2009}.

Moreover, the number of input gates appearing in $C$ is 
$\sum_{i=1}^{\alpha}|A_i|+|B_i|$ and so finding a f-rep of
$\Hom(\graphK_2,\graphG)$ of minimum size corresponds to minimising the sum of the sizes of
the partitions in our complete bipartite covering of $\graphG$, call this
the cost of the covering. Proving Theorem~\ref{thm: f-clique} for the
case $k=2$,
 corresponds to finding graphs where every
covering of the edges by complete bipartite subgraphs has high
cost. This is a problem investigated by Chung et al. in
\cite{Chung1983}, where one key idea is that if a graph contains no
large complete bipartite subgraphs and a large number of edges then the cost of any cover must be high. We deploy this
idea in our more general context. This motivates the following 
lemma, which follows from a simple probabilistic
argument. 

\begin{lemma} \label{thm: G}
  For every $k \in \mathbb{N}$ there exists some $c_k \in \mathbb{R}^{+}$ such that for every sufficiently large integer $n$ there is a
  graph $\graphG$ with $n$ vertices, such that
  \begin{enumerate}
  \item $\graphG$ has $m \ge \frac18 n^2$ edges, \label{label:manyedges}
  \item $\graphG$ contains no complete bipartite
      subgraph $\graphK_{a,a}$ for $a \geq3\log(n)$, and \label{label:nobipart}
  \item the number of $k$-cliques in $\graphG$ is at least $c_kn^{k}$. \label{label:manycliques}
    \end{enumerate}
  \end{lemma}

\begin{proof}
We first prove the following claim.

\begin{claim} \label{thm: bipart}
Let $\graphG_n$ be a random graph on $n$ vertices with edge probability
$\frac{1}{2}$. Let $\epsilon >0$. Then for any $a = a(n) \geq (2+\epsilon)\log(n)$,
\[
 P_a:= \mathbb{P}(\graphG_n \textrm{ has } \graphK_{a,a} \textrm{ as a subgraph}) \to
   0 \textrm{ as } n \to \infty.
 \]
  
\end{claim}

\begin{proof}[Proof of Claim] By the union bound and the bound on $a$ we get 
  \begin{equation*}
    \label{eq:2}  
    P_a  \le  \binom{n}{a}^2 2^{-a^2} 
    \le n^{2a}2^{-a^2}  = 2^{2a\log(n)-a^2} \leq
    2^{-(\epsilon^2+2\epsilon)\log^2n}.
  \end{equation*} 

\end{proof}
    Now let $\graphG_n$ be as above, $s = s(k) := \binom{k}{2}+1$ and 
     $p$ be the probability that such a graph has at
    least $\binom{n}{k}2^{-s}
    $ $k$-cliques. The expected number of
    $k$-cliques in $\graphG_n$ is $\binom{n}{k}2^{-\binom{k}{2}}$. Therefore,
  \begin{equation*}
      \label{eq:3}
   \binom{n}{k}2^{-\binom{k}{2}} \le \binom{n}{k}2^{-s} (1-p) + \binom{n}{k}p
    \end{equation*}
and so $p \ge 1/(2^{s}-1)$. Moreover, by the Chernoff bound, (1) from the statement of the Lemma fails only
  with exponentially small probability. By Claim~\ref{thm: bipart} 
  there must exist a $\graphG$ satisfying
(1), (2), and (3) for sufficiently large $n$.
\end{proof} 
\noindent
  Equipped with Lemma~{\ref{thm: G}} we are already in a position to prove  Theorem~\ref{thm: f-clique}.

\begin{proof}[Proof of Theorem~\ref{thm: f-clique}]

Let $\graphG$ be an $n$-vertex graph provided by Lemma~\ref{thm: G} and suppose that
$C$ is a f-rep for $\graphK_k$ and $\graphG$. If $\max_{x\in\dom(g)}|\{a\mid h(x)=a, h\in S_g\}| \le 3\log(n)$ for a gate $g$ we say that $g$ is \emph{small}. Otherwise we say $g$ is \emph{big}. Note that a $\times$-gate cannot have two big children
$g_1$ and $g_2$ because otherwise there would be $x_1\in\dom(g_1)$
and $x_2\in\dom(g_2)$ such that 
\[
\{a\mid h(x_1)=a, h\in S_{g_1}\} \times \{a\mid h(x_2)=a, h\in S_{g_2}\}
\]
forms a complete bipartite subgraph
with partitions bigger than $3\log n$ in $\graphG$, contradicting \eqref{label:nobipart} from Lemma~\ref{thm: G}.
 
If $g$ is 
small, then $C_g$ represents $|S_g|\leq 3^{|\!\dom(g)|}\log^{|\!\dom(g)|}(n)$
homomorphisms. We claim that for any gate $g$ of $C$, $|S_g| \le
|C_g| \cdot 3^{|\!\dom(g)|}\log^{|\!\dom(g)|}(n)$. Clearly this holds for input gates. We can
therefore induct bottom up on $C$. Suppose our claim holds for all children $g_1, \dots, g_r$ of
some gate $g$. 
If $g$ is a $\times$-gate then we know
at most one of the $g_i$ is big, say $g_1$. Define $b:=\sum_{i=2}^r |\!\dom(g_i)|$. Then,
\[ |S_g| = \prod_{i=1}^{r} |S_{g_i}| \le
  |C_{g_1}|\cdot 3^{|\!\dom(g_1)|}\log^{|\!\dom(g_1)|}(n) \cdot
  3^b\log^b(n) \le |C_g|\cdot 3^{|\!\dom(g)|} \log^{|\!\dom(g)|}(n),
\]
 The $\cup$-gate case follows
immediately from the induction hypothesis because the circuit is treelike so if $g$ has children $g_1, \dots, g_r$ then $|C_g| = 1+ \sum_{i=1}^r |C_{g_i}|$.  

From the claim we infer in particular that $|\!\Hom(\graphK_k,\graphG)| = |S_s| \le
|C| \cdot 3^k\log^k(n)$ for the sink $s$ of $C$. By
\eqref{label:manycliques} from Lemma~\ref{thm:
  G} it follows that
$|C| \ge c_kn^k /(3^k\log^k(n)) $ which, combined with \eqref{label:manyedges} from Lemma~\ref{thm: G}, implies the claimed result.         
 \end{proof}

  We now transfer this bound to d-reps, by showing that, for the same
  graphs used above, any d-rep cannot be much smaller than the smallest
  f-rep. 

\begin{proof}[Proof of Theorem \ref{thm: d-clique}]

  Let $\graphG$ be an $n$-vertex graph provided by Lemma~\ref{thm: G} as above and $C$ a d-rep of
  $\Hom(\graphK_k,\graphG)$ with sink $s$. If a gate has
  out-degree of more than one we call it a \emph{definition}. As in the proof of Theorem~\ref{thm: f-clique}, if $\max_{x\in\dom(g)}|\{a\mid h(x)=a, h\in S_g\}| \le 3\log(n)$ for a gate $g$ we say that $g$ is \emph{small}. Otherwise we say $g$ is \emph{big}.

Our strategy is to convert $C$ into an equivalent f-rep that is not much bigger than $C$. 
For ease of analysis and exposition we will do this by first eliminating all small 
definitions and then all big definitions. First if $s$ is small
replace the whole circuit with its equivalent flat
representation. Otherwise, we mark all small gates $g$ that have a big
parent and compute the equivalent flat representation $F_g$ of
$C_g$. Since all umarked small gates are descendant of some marked
gate, we can now safely delete them. Afterwards we consider every wire
between a marked gate $g$ and one of its big parents $p$ and replace
it by a copy of $F_g$ as input to $p$. We obtain an equivalent circuit $\hat{C}$
where every small gate has only one parent. The size (number of
gates plus number of wires) increases only by a factor determined by
the maximum size of a flat representation:

\begin{claim} \label{claim: small}
$\|\hat{C}\| \le \|C\| \cdot (2k+3)3^k\log^{k}(n)$.%
\end{claim}

When we try and eliminate big definitions one challenge is that if $g$ is big, then $\|
\hat{C}_g\|$ can be large and so making lots of copies of it could blow up the size of our 
circuit. To overcome this we introduce the notion of an \emph{active parent}. We then show that non-active parents are effectively redundant and that there can't be too many active ones, which allows us to construct an equivalent treelike circuit of the appropriate size.

So let $g$ be a definition with parents $p_1, \dots, p_{\alpha}$, $\alpha > 1$, and 
suppose there is a unique path from $p_i$ to the sink $s$ for every $i$. Then for every 
gate $v$ on the unique path from $g$ to $s$ which passes through
$p_i$, we inductively define a set 
of (partial) homomorphisms $A_i^v=A_i^v(g)$ as follows, where $\hat{v}$ refers to the 
child of $v$ also lying on this path. 
\begin{itemize}
\item $A^{g}_i := S_g$,
\item if $v$ is a $\cup$-gate $A^v_i:= A_i^{\hat{v}}$,
\item otherwise $v$ is a $\times$-gate with children $u_1, \dots,
  u_{r-1}, \hat{v} $ and $A^v_i := $  \\
$  \{h_1 \cup \ldots \cup h_r \, \mid \, h_i \in S_{u_i}, \, i \in
  [r-1] ,\, h_r \in A_i^{\hat{v}} \big  \}. $
\end{itemize}
Write $A_i := A_i^s$, intuitively this is the set of homomorphisms that the wire from $g$ to $p_i$
contributes to. We say that a parent $p_i$ of $g$ is \emph{active} if $A_i \nsubseteq \cup_{j \neq
  i}A_j$. Now using a top-down traversal starting at the output gate of $\hat{C}$, we replace each 
gate with active parents $p_1, \dots, p_{\beta}$, by $\beta$ copies $g_1, \dots, g_{\beta}$ 
such that the children of each $g_i$ are exactly the children of $g$ and $g_i$ has 
exactly one out-edge going to $p_i$. At each stage we also clean-up the circuit by 
iteratively deleting all gates which have no incoming wires, as well as all the wires 
originating from such gates. We can think of this process as constructing a slimmed 
down version of the traversal, where at each stage we only keep wires going to active 
parents. Call the resulting circuit $C'$. 

We first note that this process is well-defined, as there is a unique
path from the sink to itself and since whenever 
we visit a gate we have already visited all of its parents. Moreover, by construction this 
results in a treelike circuit. In the next claim we bound the size of $C'$ and show it is indeed an equivalent circuit. The idea is that firstly a gate cannot have too many 
active parents, as otherwise we would get a large biclique in $\graphG$ which is ruled out by Lemma~\ref{thm: G}, and secondly that since only active parents contribute new homomorphisms we really do get an equivalent circuit, \confORfull{see the full version of this paper for details.}{see Appendix~\ref{ap: d-clique} for details.}{ see Appendix~\ref{ap: d-clique} for details.}

\begin{claim} \label{claim: big}
$C'$ is a f-rep of $\Hom(\graphK_k, \graphG)$ and $\|C'\| \le 3^{k}\log^{k-1}(n)\|\hat{C}\|$.
\end{claim}  
Pulling everything together we get that 
\[
\|C\| \;\underset{\text{\footnotesize (Claim~\ref{claim: small})}}{\ge}\; \frac{\|\hat{C}\|}{(2k+3)3^k\log^{k}(n)} \;\underset{\text{\footnotesize (Claim~\ref{claim: big})}}{\ge}\; \frac{\|C'\|}{(2k+3)3^{2k}\log^{2k-1}(n)} = \Omega \left(\frac{m^{k/2}}{\log^{3k-1}(m)}\right),
\]
where the final equality follows by Theorem~\ref{thm: f-clique} since $C'$ is a f-rep of $\Hom(\graphK_k,\graphG)$.
\end{proof}

\section{The representation dichotomy for structures of bounded
  arity} \label{s: dichotomy}
 
In this section we lift the lower bound for cliques to all classes of graphs with
unbounded treewidth. We first introduce a notion of reductions between representations and
show that having lower bounds for all graphs of unbounded treewidth
immediately implies our main dichotomy theorem for bounded-arity structures.

Afterwards, we introduce minor and almost-minor reductions and use
them to obtain a lower bound for representing homomorphisms from large
grids and from graphs having large grids as a minor. The
superpolynomial representation lower bound for all graph classes with
unbounded treewidth then follows from the excluded grid theorem.

\subsection{Reductions between representations} \label{ss: reduct}

In order to define reductions between representations we fix some
notation. For two structures $\strucA$ and $\strucB$ we let
$\dset(\strucA,\strucB)$ be the set of all d-representations of
$\Hom(\strucA,\strucB)$ and $\dsize(\strucA,\strucB) = \min_{C\in
  \dset(\strucA,\strucB)}\|C\|$ be the size of the smallest such
representation.

For a class $\classC$ of structures the function $\dsub{\strucA,\classC}\colon
\mathbb N\to \mathbb N$ expresses the required size of a
d-representation of homomorphisms between $\strucA$ and
$\strucC\in\classC$ in terms of the size $m$ of $\strucC$, i.\,e.,
$\dsub{\strucA,\classC}(m) = \max_{\{\strucC\in\classC \,\colon 
  \|\strucC\|\leq m\}}\dsize(\strucA,\strucC)$. We write
$\dsub{\strucA}$ as an abbreviation for $\dsub{\strucA,\classC}$ when
$\classC$ is the class of all structures.
Translated to this notation,  \cite{olteanu2015}
showed that $\dsub{\strucA}= O(m^{\tw(\strucA)+1})$, whereas
Theorem~\ref{thm: d-clique} states the lower bound $\dsub{K_k} =
\Omega(m^{k/2}/\log^{3k-1}(m))$. We also write, for a signature $\sigma$, $\classC_{\sigma}$ to denote the class of all $\sigma$-structures.

The main goal of this section is to prove, for some increasing function
$f$, a lower bound of the form
$\dsub{\strucA}= \Omega(m^{f(\tw(\strucA))/\ar(\strucA)})$ for every
structure $\strucA$, which immediately implies our main theorem.
To achieve this we use reductions with our $k$-clique lower
bound as a starting point. Suppose we already have a lower bound on
$\dsub{\strucA,\classC}$ for a class $\classC$ of arbitrarily large
hard instances (implying a lower bound on $\dsub{\strucA}$), then we can
use the following reduction \emph{from $\strucA$ to $\strucB$ via
  $\classC$} to obtain a lower bound on $\dsub{\strucB}$.

  \begin{definition} \label{def:reduction}
    Let $\strucA$ be a $\sigma$-structure and let
    $\classC$ be a class of $\sigma$-structures. Let $\strucB$ be a $\sigma'$-structure and $c \colon
    \mathbb{R}^{+} \to \mathbb{R}^{+}$ be a
    strictly increasing function. Then a \emph{$c$-reduction} from $\strucA$ to $\strucB$
    via $\classC$ is a pair $(\phi,
    (\psi_{\strucC})_{\strucC \in \classC})$, where $\phi \colon
    \classC \to \classC_{\sigma'}$ and $\psi_{\strucC} \colon
    \dset(\strucB,\phi(\strucC)) \to \dset(\strucA,\strucC)$
    such that:
    \begin{enumerate}
    \item for every $n\in \mathbb N$ there is a $\strucC\in\classC$
      such that $\|\phi(\strucC)\| \geq n$,
    \item $\|\phi(\strucC)\| \leq c(\|\strucC\|)$ for all $\strucC\in\classC$, and  
      \item $\|\psi_{\strucC}(C)\| \leq \|C\|$ for every structure
          $\strucC \in \classC$ and circuit $C \in \dset(\strucB,
        \phi(\strucC))$.
      \end{enumerate}
      If $c(m) = \alpha m$ for some $\alpha \in \mathbb{R}^{+}$, we say we have
      a \emph{linear reduction}. 
    \end{definition}
        
    \begin{lemma}\label{lem:reduct}
      Suppose there is a $c$-reduction $(\phi,
    (\psi_{\strucC})_{\strucC \in \classC})$ from $\strucA$ to $\strucB$ via
    $\classC$, let $\classD = \{\phi(\strucC) \, \mid \,
    \strucC\in\classC\}$ be the image of $\phi$. Then
    $
    \dsub{\strucB,\classD} = \Omega(\dsub{\strucA,\classC}\circ\lfloor c^{-1}\rfloor).
    $
    \end{lemma}

    \begin{proof}
     Fix $m \in \mathbb{N}$, where $m \ge \min_{\{\strucC \in \classC\}}\|\strucC\|$. Let $\strucC \in \classC$ with $\|\strucC\| \leq m.$
     Then $\psi_{\strucC}$
     witnesses that $\dsize(\strucA, \strucC) \leq  \dsize(\strucB, \phi(\strucC))$. Also $\|\phi(\strucC)\| \leq c(m)$, since $c$ is an
     increasing function. So 
     $\dsub{\strucB,\classD}(c(m)) = \max_{\{\strucC \, \colon\|\phi(\strucC)\|
         \leq c(m)\} } \dsize(\mathcal{B},
       \phi(\strucC)) \geq \max_{\{\strucC \, \colon \|\strucC\| \leq m\}} \dsize(\strucA,
       \strucC) = \dsub{\strucA,\classC}(m).
     $
     Since $\classC$ and $\classD$ contain arbitrarily large structures, the asymptotic bound from the lemma follows.
    \end{proof}

We start illustrating the power of these reductions by making two
simplifications. First, we reduce the general problem of representing
homomorphisms to representing homomorphisms that respect a
partition. Second, we further reduce to graph homomorphisms that
respect a partition. \confORfull{All proofs from this subsection can be found in the full version of the paper.}{}{All proofs from this subsection can be found in Appendix~\ref{ap: dicot}.}

For the first reduction we need the notion of the \emph{individualisation} of a $\sigma$-structure $\strucA$, which is obtained from $\strucA$ by giving every element of the universe a distinct color. More precisely, we extend the vocabulary $\sigma$ with 
unary relations (= colours) $\sigma_A=\{P_a\colon a\in A\}$ and
let $ \strucA^{\text{id}}$ be the $\sigma\cup\sigma_{A}$-expansion of
$\strucA$ by adding $P^{\strucA^{\text{id}}}_a= \{a\}$.

\begin{lemma}\label{lem:id-reduct}
Let $\strucA$ be a $\sigma$-structure and let $\classC$ be the class
of all $\sigma \cup \sigma_{A}$-structures where \\ $\{P^{\strucC}_a\, \mid \, a\in A\}$ is a partition of the universe. Then
$
\dsize_{\strucA} = \Omega(\dsize_{\strucA^{\text{id}}, \classC}).
$
\end{lemma}

\confORfull{}{%
\begin{proof}

 We give a linear reduction from $\strucA^{\text{id}}$ to $\strucA$  
  via $\classC$ and then conclude via Lemma~\ref{lem:reduct}.
For $\strucC \in \classC$ we define
$\phi(\strucC)$ to be the $\sigma$-reduct of $\strucC$ and for a
representation $C \in \dset(\strucA, \phi(\strucC))$   we let $\psi_{\strucC}(C) = C'$, where $C'$ is a factorised circuit with 
\[
S_{C'} = \{h \in \Hom(\strucA, \phi(\strucC)) \, \mid \, h(a) \in P^{\strucC}_a \text{ for all } a \in A \}
\] 
and $\|C'\| \le \|C\|$, the existence of which is guaranteed by Lemma~\ref{lem:technical}.

We claim this defines a reduction from $\strucA^{\text{id}}$ to $\strucA$
via $\classC$. We just need to check that $\psi_{\strucC}(C)=C' \in \dset(\strucA^{\text{id}}, \strucC)$, for every $\strucC \in \classC$ and $C \in \dset(\strucA, \phi(\strucC))$. To do this we show that $\Hom(\strucA^{\text{id}}, \strucC)= S_{C'}$. So let $h \in \Hom(\strucA^{\text{id}}, \strucC)$, then $h(a) \in P^{\strucC}_a$ for all $a$, since $P^{\strucA^{\text{id}}}_a = \{a\}$. Moreover since $\phi(\strucC)$ is the $\sigma$-reduct of $\strucC$ it is clear that $h \in \Hom(\strucA, \phi(\strucC))$, which implies $h\in S_{C'}$. Conversely, if $h\in S_{C'}$ it is a map from $\strucA^{\text{id}} \to \strucC$, that respects all relations in $\sigma \cup \sigma_{A}$ and so we do have the claimed reduction. 
\end{proof} %
}

We call structures and (vertex-coloured) graphs
\emph{individualised} if every vertex has a distinct colour.
In the next lemma we reduce from individualised structures to
individualised graphs. Recall the definition of the Gaifman graph
$\graphG_{\strucA}$ from the preliminaries. 

\begin{lemma} \label{lem:Gaifman-reduct}
  Let $\strucA$ be an individualised structure and
  $\graphG^{\text{id}}_\strucA$ the individualisation of its Gaifman
  graph. Let $\classC$ be the class of all structures $\strucC$ where
  $\{P^{\strucC}_a \, \mid \, a\in A\}$ is a partition of its universe and 
$\classH$ be the class of all
  vertex-coloured graphs $\graphH$ where
  $\{P^{\graphH}_a \, \mid \, a\in A\}$ is a
  partition of its vertex set.
  Then
  $\dsub{\strucA,\classC}(m) = \Omega\big(\,(\dsub{\graphG^{\text{id}}_{\strucA},\classH}(m))^{2/\ar(\strucA)}\,\big)$.
\end{lemma}

\confORfull{}{%

\begin{proof}
 Let $\graphH \in \classH$ and suppose $\strucA$ is an individualised $\sigma$-structure. We define $\phi(\graphH) \in \classC$ as follows. The universe of $\phi(\graphH)$ is $V(\graphH)$ and each element is coloured in the same way as in $\graphH$. Then we define for $R \in \sigma$ of arity $t$
\[
R^{\phi(\graphH)} = \{ (x_1, \dots, x_t) \, \mid \, \{x_i,x_j\} \in E(\graphH) \text{ whenever } x_i \neq x_j \}.
\]
If a tuple in  $R^{\phi(\graphH)}$ has $i > 1$ distinct coordinates this corresponds to an $i$-clique in $\graphH$. Since there are at most $i^{t}$ tuples in $R^{\phi(\graphH)}$ corresponding to each $i$-clique in $\graphH$. As the number of $i$-cliques in a graph with $m$ edges is upper bounded by $m^{i/2}$, one can see that
\[
|R^{\phi(\graphH)}| \le |V(\graphH)| + t^t \cdot \sum_{i=2}^t \|\graphH\|^{i/2} = O(\|\graphH\| + \|\graphH\|^{t/2}).
\]
 Where for the last equality recall that $\graphH$ is a vertex coloured graph, so $\|\graphH\| \ge |V(\graphH)|$. Since the arity of each relation in $\sigma$ is bounded by $r:=\ar(\strucA)$ it follows that $\|\phi(\graphH)\| = O(\|\graphH\|^{r/2})$.

We next claim that $\Hom(\strucA, \phi(\graphH)) = \Hom(\graphG_{\strucA}^{\text{id}}, \graphH)$. First assume the claim.  Set $\psi_{\graphH}(C) = C$ and observe that since $\classC$ contains structures of unbounded size and $\|\phi(\graphH)\| \ge \|\graphH\|$ for every $\graphH \in \classC$, we have that $\phi(\classC)$ contains structures of unbounded size. We therefore have a $O(m^{r/2})$ reduction from $\graphG_{\strucA}^{\text{id}}$ to $\strucA$ via $\classH$ and so by Lemma~\ref{lem:reduct} the result follows, noting again that $\phi(\graphH) \in \classC$ for every $\graphH \in \classH$.

We now prove the claim. Let $h \in \Hom(\strucA, \phi(\graphH))$ and suppose $\{u,v\} \in 
E(\graphG_{\strucA^{\text{id}}})$. Then $u\neq v$ by the definition of the Gaifman graph and there is some relation $R \in \sigma$ and $(a_1,
\dots,a_t) \in R^{\strucA}$ such that $u=a_i$, $v=a_j$, for some $i,j$. Since $h \in \Hom(\strucA, \phi(\graphH))$,
$(h(a_1),\dots,h(a_t)) \in R^{\phi(\graphH)}$. Moreover, as $h$   preserves the relations $P_a$ and $\{P_a^{\phi(\graphH)}\}_{a\in A}$  partitions the universe of $\phi(\graphH)$, $h(u) \neq h(v)$. By the definition of $R^{\phi(\graphH)}$ it 
follows that $\{h(u),h(v)\} \in E(\graphH)$ and so 
 $h \in \Hom(\graphG_{\strucA}^{\text{id}}, \graphH)$. Conversely suppose $h \in 
\Hom(\graphG_{\strucA}^{\text{id}}, \graphH)$. Let $(a_1, \dots, a_t) \in R^{\strucA}$ for 
some $R \in \sigma$. By the definition of the Gaifman graph $\{a_i,a_j\} \in 
E(\graphG_{\strucA}^{\text{id}})$ for every $a_i \neq a_j$. Therefore, since $h \in 
\Hom(\graphG_{\strucA}^{\text{id}}, \graphH), \{h(a_i), h(a_j)\} \in E(\graphH)$ for every $a_i  
\neq a_j$. By construction $(h(a_1),\dots, h(a_t)) \in R^{\phi(\graphH)}$ and we may 
infer that $h \in \Hom(\strucA, \phi(\graphH))$. 
\end{proof}

}

Taking both lemmas into account, we can now focus on individualised graphs $\graphG$ on
the left-hand side and on graphs $\graphH$ with the corresponding colouring
$\{P^{\graphH}_a \, \mid \, a\in V(\graphG)\}$ that partitions 
its vertex set on the right-hand side. We
call such graphs \emph{$V(\graphG)$-partitioned graphs}. However we would also like to deploy our lower bound from Section~\ref{s: clique}; the next lemma allows to transfer this lower bound to individualised structures.

\begin{lemma} \label{lem:reductgraph}
Let $\graphG$ be a graph and $\classC$ be the class of all $V(\graphG)$-partitioned graphs.
Then $
\dsize_{\graphG^{\text{id}},\classC} = \Omega(\dsize_{\graphG}).
$
\end{lemma}

\confORfull{}{%
\begin{proof}
Let $\graphH$ be a graph. We define $\phi(\graphH)$ to be the following $V(\graphG)$-partitioned graph:
\begin{align*}
 V(\phi(\graphH)) &= \{v_a \, \mid \, v \in V(\graphH), \, a \in V(\graphG) \},\\
 E(\phi(\graphH)) &= \left\{\{v_a,u_b\} \, \mid \, \{v,u\} \in E(\graphH), \{a,b\} \in E(\graphG)\right\}, \\
 P_a^{\graphH} &= \{v_a \, \mid \, v \in V(\graphH) \}, \text{ for every $a \in V(\graphG)$}.
\end{align*}
Recall from the preliminaries that all structures are assumed to be  connected. In particular, since $\graphG$ and $\graphH$ are both connected, $\phi(\graphH)$ is connected so $|V(\phi(\graphH))| \le 1/2|E(\phi(\graphH))| \le 1/2|E(\graphG)|\cdot|E(\graphH)|$ we have that
\[ \|\phi(\graphH)\| = |V(\phi(\graphH))| + |E(\phi(\graphH))| \le |V(\phi(\graphH))|+|E(\graphG) |\cdot |E(\graphH)|\le 2|E(\graphG)|\cdot|E(\graphH)|.
\]
Let $C \in \dset(\graphG^{\text{id}}, \phi(\graphH))$. Relabel every input gate of the 
form $(a \mapsto v_a)$ by $(a \mapsto v)$, then the resulting circuit, $C'$ is the same size as $C$ and we claim it is a d-rep of $\Hom(\graphG, \graphH)$. 

 To see this let $h\in \Hom(\graphG, H)$. For every $\{a,b\} \in E(\graphG)$, $\{h(a),h(b)\} \in E(\graphH)$ and so  
$\{h(a)_{a}, h(b)_{b}\} \in E(\phi(\graphH))$. Therefore the map $h' 
\colon \graphG^{\text{id}} \to \phi(\graphH)$, such that $h'(x)=h(x)_x$ is a 
homomorphism. Similarly for every $h \in \Hom(\graphG^{\text{id}}, \phi(\graphH))$, 
$gh \in \Hom(\graphG, \graphH)$, where $g \colon \phi(\graphH) \to \graphH$ with $g(v_a)=v$. We can infer that $C'$ is indeed a a d-rep of $\Hom(\graphG, \graphH)$ and so we have a linear
reduction from $\graphG$ to $\graphG^{\text{id}}$ via the class of all graphs. By 
Lemma~\ref{lem:reduct} the result follows. 
\end{proof} %
}

\subsection{Minor reductions}

 In this subsection we show that we can reduce $\graphG'$ to $\graphG$
 if $\graphG$ is a minor of $\graphG'$. 
 We start by illustrating how to handle edge
 contractions via an example.

 \begin{example}[Reduction from 4-cycle to 3-cycle] \label{ex: 3to4}
   Consider the 3-cycle $\graphK_3$ on vertices $x_1,x_2,x_3$, which is a
   minor of the 4-cycle $\graphC_4$ on vertices $x_1,x_2,x_3,x_4$ by
   contracting one edge $\{x_4,x_1\}$. We show that we can lift the
   lower bound for $\graphK^{\text{id}}_3$ (Theorem~\ref{thm: d-clique} + Lemma~\ref{lem:reductgraph})
   to $\graphC^{\text{id}}_4$ (and hence $\graphC_4$ by Lemma~\ref{lem:id-reduct}) by a simple
   linear reduction from $\graphK^{\text{id}}_3$ to $\graphC^{\text{id}}_4$ via the class of all
   $\{x_1,x_2,x_3\}$-partitioned graphs.
   Let $\graphH$ be a
   $\{x_1,x_2,x_3\}$-partitioned graph. We define the $\{x_1,x_2,x_3,x_4\}$-partitioned graph
   $\graphH'=\phi(\graphH)$
   by
 $P^{\graphH'}_{x}:= P^{\graphH}_{x}$ for $x\in\{x_1,x_2,x_3\}$,
   $P^{\graphH'}_{x_4}:= \{\widehat{v}\mid v\in P^{\graphH}_{x_1}\}$ and $E(\graphH')=$
 \begin{align*}  &\big\{\{v,\widehat{v}\} \mid  v\in P^{\graphH}_{x_1} \big\} \\
 \cup \,&\big\{\{v,w\}  \mid  v\in P^{\graphH}_{x_1},\, w\in P^{\graphH}_{x_2},\,\{v,w\}\in E(\graphH) \big\} \\
   \cup\, &\big\{\{v,w\}  \mid  v\in P^{\graphH}_{x_2},\, w\in P^{\graphH}_{x_3},\, \{v,w\}\in E(\graphH) \big\} \\
  \cup\, &\big\{\{v,\widehat{w}\}  \mid  v\in P^{\graphH}_{x_3},\, w\in P^{\graphH}_{x_1},\,\{v,w\}\in E(\graphH) \big\}.
   \end{align*}
   Note that the size of $\graphH'$
   is linear in the size of $\graphH$.
   The construction ensures that any mapping
   $h'\colon\{x_1,\ldots,x_4\}\to V(\graphH')$ is a
   homomorphism from $\graphC^{\text{id}}_4$ to $\graphH'$ if, and only if, $h'(x_4) =
   \widehat{h'(x_1)}$ and $h(x_i):=h'(x_i)$, for $i\in[3]$, is a homomorphism from
   $\graphK^{\text{id}}_3$ to $\graphH$. Therefore, $\Hom(\graphK^{\text{id}}_3,\graphH) =
   \pi_{\{x_1,x_2,x_3\}}\Hom(\graphC^{\text{id}}_4,\graphH')$ and a
   representation $C'$ of $\Hom(\graphK^{\text{id}}_3,\graphH)$ can be obtained from a
   representation $C$ of $\Hom(\graphC^{\text{id}}_4,\graphH')$ by
   Lemma~\ref{lem:technical} which, moreover, guarantees that $\|C'\|
   \le \|C\|$. Therefore we do have a linear reduction from
   $\graphC^{\text{id}}_4$ to $\graphK^{\text{id}}_3$. It follows that $\dsize_{\strucC_4}(m) = \Omega(\dsize_{\strucC_4^{\text{id}}, \classC}(m))
  =\Omega(\dsize_{\graphK_3^{\text{id}},\classH}(m)) = \Omega(\dsize_{\graphK_3}(m))=\Omega(m^{3/2}/ \log^7(m))$, where $\classC$ is the class of $V(\graphC_4^{\text{id}})$-partitioned graphs and $\classH$ is the class of $V(\graphK_3^{\text{id}})$-partitioned graphs. The first equality follows by Lemma~\ref{lem:id-reduct}, the second by Lemma~\ref{lem:reduct}, the third by Lemma~\ref{lem:reductgraph} and the last by Theorem~\ref{thm: d-clique}. 
 \end{example}

So to handle edge contractions we take the
 partitioned hard right-hand side instance and
 ``re-introduce'' the edge $\{x,y\}$ contracted to $x$ by copying
 $P_x$ to $P_y$ and adding a perfect matching
 between the two partitions $P_x$ and $P_y$. Handling edge
 deletions is even simpler: suppose that $\{x,y\}$ is deleted from
 $\graphG'$ to $\graphG$ and we want to reduce $\graphG'$ to $\graphG$. Then we take a
 partitioned hard instance for $\graphG$ and just introduce the complete
 bipartite graph between the partitions $P_x$ and $P_y$; this may
 square the size of the graph. Since the sets of (partition-respecting)
 homomorphisms are the same for both instances, we do not even have to
 modify the representations in the reduction.
 The next lemma summarises these findings. Its proof is omitted as it
 is subsumed by Lemma~\ref{lem:almost-minor-red}.

 \begin{lemma}\label{lem:minor-reduction}
   Let $\graphG_X, \graphG_Y$ be graphs with vertex sets $X$ and $Y$ respectively such that $\graphG_X$ is a minor of $\graphG_Y$. Let $\classH$ be the class of all
   $V(\graphG_X)$-partitioned graphs and $\classH'$ the class of all
   $V(\graphG_Y)$-partitioned graphs.
   Then there is a $c$-reduction
   $(\phi,(\psi_\graphH)_{\graphH\in\classH})$ from
   $\graphG_Y^{\text{id}}$ to $\graphG_X^{\text{id}}$
    via $\classH$ with $\phi(\classH)\subseteq \classH'$ and $c(m)=O(m^2)$.
 \end{lemma}

 This yields together with Lemmas~\ref{lem:reduct}, \ref{lem:id-reduct} and \ref{lem:reductgraph} along with Theorem~\ref{thm: d-clique} the following corollary.

 \begin{corollary}\label{cor:Kk-minor}
   If $\graphG$ has $\graphK_k$ as a minor, then $\dsub{\graphG}=\Omega(m^{k/4}/\log^{(3k-1)/2}(m))$.
 \end{corollary}

\subsection{Relaxation of the minor condition} \label{ss: am}

Every graph having $\graphK_k$ as a minor has treewidth at least $k-1$, so
Corollary~\ref{cor:Kk-minor} provides the desired lower bound of Theorem~\ref{thm:maintheorem} for certain
large-treewidth graphs. However, there are graphs of large treewidth
that do not have a large clique as a minor. Instead, the excluded grid
theorem \cite{Robertson1986} and its more efficient version \cite{DBLP:journals/jacm/ChekuriC16} tells us that
graphs of large treewidth always have a large $k\times k$-grid as a
minor.
  \begin{theorem}[\cite{DBLP:journals/jacm/ChekuriC16}] \label{thm: EG}
    There is a polynomial function $w:\mathbb N\to\mathbb N$ such that for every $k$ the $(k \times k
    )$-grid is a minor of every graph of treewidth at least $w(k)$. 
  \end{theorem}
Thus, in order to prove Theorem~\ref{thm:maintheorem} it suffices to combine
 Lemma~\ref{lem:minor-reduction} with
a lower bound for grid graphs. We cannot reduce immediately to our $k$-clique lower bound, as the grid does not have a $\graphK_k$
minor for $k\geq 5$. However, the complete graph $\graphK_k$
is ``almost a minor'' of $\graphG_{2k-2}$ for the following notion of
\emph{almost minor} that is good enough to prove a variant of
Lemma~\ref{lem:minor-reduction}.

 \begin{definition} \label{def: am}
   For two graphs $\graphG_X$, $\graphG_Y$ with vertex sets
   $X=V(\graphG_X)$ and $Y=V(\graphG_Y)$ %
   we say that a map
   $M \colon Y \to 2^{X}$ is \emph{almost minor} if the following
   conditions hold:
   \begin{enumerate}
     \item for every $y \in Y$, $|M(y)| \in \{1,2\}$;
   \item for every $x\in X$ there is a $y\in Y$ s.t.\ $M(y)=\{x\}$ 
 and for every $x,x'$ adjacent in $\graphG_X$ there exists $y,y'$ adjacent
        in $\graphG_Y$ such that $M(y) = \{x\}$ and $M(y') = \{x'\}$;
        \item for each $x \in X$, $ \{ y \colon x \in M(y)  \}$ is
          connected in $\graphG_Y$ and 
          \item if $M(y)=\{x,x'\}$ with $ x \neq x'$ and $y'$ is
            adjacent to $y$ in $\graphG_Y$, then $M(y') = \{x\}$ or  $M(y') =
            \{x'\}$.
        \end{enumerate}
       If such a map exists we say $\graphG_X$ is an almost minor of
       $\graphG_Y$.
     \end{definition}
 
For the special case when $|M(y)|=1$ for all $y$, $M$ is a
\emph{minor map} and $G_X$ is a minor of $G_Y$. The motivation for 
this definition is that whilst grids are planar, large cliques are 
not and so we introduce ``junctions'', i.e. nodes $y$ such that 
$M(y) = \{x_1, x_2\}$ which allows $\{v \, \mid \, x_i \in M(v)\}$, 
$i\in \{1,2\}$ to intersect in a controlled way, see 
Figure~\ref{fig: gridcon}. We should also observe here that this 
notion is related to Marx's notion of an \emph{embedding} 
\cite{marx2007}.\footnote{In particular the definition of a 
\emph{depth-2 embedding} can be obtained from our definition of an 
almost minor by the following modifications. First remove clause 
(4). Second replace (2) with the following condition: for every 
$x\in X$ there is a $y\in Y$ s.t.\ $x \in M(y)$ 
 and for every $x,x'$ adjacent in $\graphG_X$ there exists 
 \emph{either} $y,y'$ adjacent
        in $\graphG_Y$ such that $x\in M(y)$ and $x' \in M(y')$ \emph{or} there exists $y$ such that $\{x,x'\} \subseteq M(y)$. If we also remove clause (1) we get the general definition of an embedding.} 
Now we can state our reduction lemma for almost minors, which extends
Lemma~\ref{lem:minor-reduction}.

     \begin{lemma} \label{lem:almost-minor-red}
       Let $\graphG_X$, $\graphG_Y$ be the graphs with vertex sets $X$
       and $Y$, respectively, such that $\graphG_X$ is an almost minor of
       $\graphG_Y$. Let $\classH$ be the class of all
       $X$-partitioned graphs and $\classH'$ be the class of all $Y$-partitioned graphs, then there is a $c$-reduction
       $(\phi,(\psi_\graphH)_{\graphH \in \classH})$ from 
       $\graphG^{\text{id}}_Y$ to $\graphG^{\text{id}}_X$ 
        via $\classH$ with $\phi(\classH) \subseteq \classH'$ and
       $c = O(m^2).$
       \end{lemma}

\begin{proof}
  We start by defining the $Y$-partitioned graph $\graphH^\ast=\phi(\graphH)$ for
  an arbitrary $X$-partitioned graph $\graphH$. To define the
  partitions, we consider two cases: if $M(y)=\{x\}$, we let
  $P^{\graphH^\ast}_y := \{v_a^{y} \,\mid \, a\in P^{\graphH}_x\}$ and if
  $M(y)=\{x,x'\}$, then $P^{\graphH^\ast}_y := \{v_{\{a,b\}}^{y}\, \mid \, a\in
  P^{\graphH}_x, b\in P^{\graphH}_{x'}\}$. 
  For every edge $\{y,y'\}\in E(\graphG_Y)$ we define the edge set
  $E_{\{y,y'\}}$ between the partitions $P^{\graphH^\ast}_y$ and
  $P^{\graphH^\ast}_{y'}$ by the following exhaustive cases:
  \begin{enumerate}
  \item\label{it:Edef1} if $M(y)=M(y')=\{x\}$: $E_{\{y,y'\}}:= \big\{\{v_a^{y},v_a^{y'}\} \, \mid \, a\in P^{\graphH}_x\big\}$
  \item\label{it:Edef2} if $M(y)=\{x\}$, $M(y')=\{x'\}$, and $\{x,x'\} \in
    E(\graphG_X)$:\\
    $E_{\{y,y'\}}:= \big\{\{v_a^{y},v_b^{y'}\} \, \mid \,
    a\in P^{\graphH}_x, b\in P^{\graphH}_{x'}, \{a,b\}\in E(\graphH)\big\}$
  \item\label{it:Edef3} if $M(y)=\{x\}$, $M(y')=\{x'\}$, $x\neq x'$, and $\{x,x'\} \notin
    E(\graphG_X)$: \\
    $E_{\{y,y'\}}:= \big\{\{v_a^{y},v_b^{y'}\}\, \mid \,
    a\in P^{\graphH}_x, \, b\in P^{\graphH}_{x'}\big\}$
  \item\label{it:Edef4} if $M(y)=\{x\}$ and $M(y')=\{x,x'\}$: $E_{\{y,y'\}}:=
    \big\{\{v_a^{y},v_{\{a,b\}}^{y'}\}\, \mid \, a\in P^{\graphH}_x, \, b\in
    P^{\graphH}_{x'}\big\}$
  \end{enumerate}
  Finally, we set $E(\graphH^\ast):= \bigcup_{e\in E(\graphG_Y)}E_e$
  and note that $\|\graphH^\ast\| = O(\|\graphH\|^2)$.
  For every homomorphism $h$ from
  $\graphG^{\text{id}}_X$ to $\graphH$ we define the mapping
  $h^\ast\colon Y\to V(\graphH^\ast)$ by
  $$
  h^\ast(y) :=
  \begin{cases}
    v^y_{h(x)},&\text{if }M(y)=\{x\}\\
    v^y_{\{h(x),h(x')\}},&\text{if }M(y)=\{x,x'\}
  \end{cases}
  $$
  The next claim provides the key property of our construction: 
  $h^\ast$ is a homomorphism from $\graphG^{\text{id}}_Y$ to
  $\graphH^\ast$ and every homomorphism from $\graphG^{\text{id}}_Y$ to
  $\graphH^\ast$ has this form\confORfull{, see the full version of this paper for a proof.}{.}{, see Appendix~\ref{ap: am} for a proof.}
  
  \begin{claim}\label{claim:HomSetsEqual}
    $\Hom(\graphG^{\text{id}}_Y,\graphH^\ast) = \{h^\ast\colon h\in \Hom(\graphG^{\text{id}}_X,\graphH)\}$
  \end{claim}
\confORfull{}{%
 \begin{proof}[Proof of Claim~\ref{claim:HomSetsEqual}]
    For the $\supseteq$-direction let $h\in
    \Hom(\graphG^{\text{id}}_X,\graphH)$. In order to
    verify that  $h^\ast$ is a homomorphism from
    $\graphG^{\text{id}}_Y$ to $\graphH^\ast$ we let $\{y,y'\}\in E(\graphG^{\text{id}}_Y)$ and
    need to check that $\{h^\ast(y),h^\ast(y')\}$ is an edge in 
    $\graphH^\ast$. By checking all cases in the definition of
    $E_{\{y,y'\}}$ it immediately follows from the definition of
    $h^\ast$ that   $\{h^\ast(y),h^\ast(y')\}\in E_{\{y,y'\}}\subseteq
    E(\graphH^\ast)$, completing one direction of the claim.

    For the $\subseteq$-direction, let $\hat h\in
    \Hom(\graphG^{\text{id}}_Y,\graphH^\ast)$. We define $Y_x := \{y\in Y \, \mid \,
    x\in M(y)\}$ for every $x\in X$ and note that this set is
    connected in $\graphG_Y$ by the definition of an almost minor
    map. We claim that there is a unique $a_x\in P_x^{\graphH}$ such that for any $y$ with $M(y)=\{x\}$, we have
    $\hat{h}(y)=v^y_{a_x}$. To see this, fix  some $y, y' \in Y_x$ and
    consider a path $y=y_0,\ldots,y_r=y'$ through the vertices in
    $Y_x$. Suppose that $\hat{h}(y_0)=v^{y_0}_{a}$, since $\hat{h}$ is a homomorphism $v^{y_0}_{a} \in P_Y^{\graphH^{\ast}}$, which implies that $a \in P_x^{\graphH}$.  By induction over the path
   we see that either $\hat{h}(y_i)=v^{y_i}_{a}$ or $\hat{h}(y_i)=v^{y_i}_{\{a,b\}}$ for
    some $b\notin P^{\graphH}_x$. This holds, since
    $\hat{h}(y_0),\ldots,\hat{h}(y_r)$ is a path in $\graphH^\ast$---as $\hat{h} \in \Hom(\graphG^{\text{id}}_Y,\graphH^\ast)$---and by points
    (\ref{it:Edef1}) and (\ref{it:Edef4}) in the definition of $E(\graphH^{\ast})$. Similarly if $M(y) = \{x, x'\}$, then $\hat{h}(y) = v^{y}_{\{a_x, a_{x'}\}}$. 
    
    Now we define $h\colon X\to V(\graphH)$ by $h(x)=a_x$ for all
    $x\in X$ and show
    that $h\in\Hom(\graphG_X^{\text{id}},\graphH)$. To see this, suppose that
    $\{x,x'\}$ is an edge in $\graphG_X$. By the definition of an
    almost minor map, there is an edge $\{y,y'\}\in E(\graphG_Y)$ with
    $M(y)=\{x\}$ and $M(y')=\{x'\}$. Therefore, $\hat{h}(y)=v^y_{a_x}$
    and $\hat{h}(y)=v^{y'}_{a_{x'}}$ are adjacent in
    $\graphH^\ast$. By point (\ref{it:Edef2}) in the definition of the
    edge set, it follows that $\{a_x,a_{x'}\}\in E(\graphH)$.
    Noting that $h^\ast = \hat{h}$ finishes the proof of
    the claim. 
  \end{proof} %
}

  We finish the lemma by
  defining the mapping $\psi_{\graphH}$ that transforms any
  d-representation for $\Hom(\graphG^{\text{id}}_X,\graphH)$ into a
  d-representation for $\Hom(\graphG^{\text{id}}_Y,\graphH^\ast)$. For
  each $x\in X$ we fix one $y_x\in Y$ such that $M(y_x)=\{x\}$ (those
  vertices exist by the definition of an almost minor map).
  Then we
  apply Lemma~\ref{lem:technical}
  and obtain a d-representation of $\pi_{\{y_x\colon x\in
    X\}}\Hom(\graphG^{\text{id}}_Y,\graphH^\ast)$. After renaming
  every $y_x$ to $x$ and every $v^y_a$ to $a$ in the input labels of
  this circuit, we get a d-representation of $\Hom(\graphG^{\text{id}}_X,\graphH)$.
\end{proof}

With the following lemma\confORfull{}{}{, whose proof is deferred to  Appendix~\ref{ap: am},} we have everything in hand to proof our main theorem. 

\begin{lemma} \label{lem:gridcon}
  For every $k$, $\graphK_k$ is an almost minor of $\graphG_{2k-2}$. 
\end{lemma}

\confORfull{\input{short_grid_proof}}{%
\begin{figure}[h]
\centering
\scalebox{0.7}{
 \begin{tikzpicture}[thick,
  every node/.style={draw,circle,minimum size=1cm, fill=black},
   ssnode/.style={fill=black},
   ]
\begin{scope}[xshift=-5cm, yshift=0cm, start chain=going below,node
  distance=8mm]
  \node[ssnode, on chain] (11) [label={[text=white] center: $1$}] {};
    \node[ssnode, on chain] (12) [label={[text=white] center: $2$}]
    {};
      \node[ssnode, on chain] (13) [label={[text=white] center: $1$}]
      {};
        \node[ssnode, on chain] (14) [label={[text=white] center: $3$}]
        {};
          \node[ssnode, on chain] (15) [label={[text=white] center:
            $1$}] {};
            \node[ssnode, on chain] (16) [label={[text=white] center: $4$}] {};
 
        \end{scope}
  
\begin{scope}[xshift=-3cm,yshift=0cm,start chain=going below,node distance=8mm]
  \node[ssnode, on chain] (21) [label={[text=white] center: $1$}] {};
    \node[ssnode, on chain] (22) [label={[text=white] center: $1, 2$}]
    {};
      \node[ssnode, on chain] (23) [label={[text=white] center: $1$}]
      {};
        \node[ssnode, on chain] (24) [label={[text=white] center: $1,3$}]
        {};
          \node[ssnode, on chain] (25) [label={[text=white] center:
            $1$}] {};
            \node[ssnode, on chain] (26) [label={[text=white] center: $1,4$}] {};
\end{scope}

\begin{scope}[xshift=-1cm,yshift=0cm,start chain=going below,node
  distance=8mm]
\node[ssnode, on chain] (31) [label={[text=white] center: $1$}] {};
    \node[ssnode, on chain] (32) [label={[text=white] center: $2$}]
    {};
      \node[ssnode, on chain] (33) [label={[text=white] center: $2$}]
      {};
        \node[ssnode, on chain] (34) [label={[text=white] center: $3$}]
        {};
          \node[ssnode, on chain] (35) [label={[text=white] center:
            $2$}] {};
            \node[ssnode, on chain] (36) [label={[text=white] center: $4$}] {};
\end{scope}

\begin{scope}[xshift=1cm,yshift=0cm,start chain=going below,node distance=8mm]
\node[ssnode, on chain] (41) [label={[text=white] center: $1$}] {};
    \node[ssnode, on chain] (42) [label={[text=white] center: $1$}]
    {};
      \node[ssnode, on chain] (43) [label={[text=white] center: $2$}]
      {};
        \node[ssnode, on chain] (44) [label={[text=white] center: $2,3$}]
        {};
          \node[ssnode, on chain] (45) [label={[text=white] center:
            $2$}] {};
            \node[ssnode, on chain] (46) [label={[text=white] center: $2,4$}] {};\end{scope}

\begin{scope}[xshift=3cm,yshift=0cm,start chain=going below,node distance=8mm]
\node[ssnode, on chain] (51) [label={[text=white] center: $1$}] {};
    \node[ssnode, on chain] (52) [label={[text=white] center: $1$}]
    {};
      \node[ssnode, on chain] (53) [label={[text=white] center: $1$}]
      {};
        \node[ssnode, on chain] (54) [label={[text=white] center: $3$}]
        {};
          \node[ssnode, on chain] (55) [label={[text=white] center:
            $3$}] {};
            \node[ssnode, on chain] (56) [label={[text=white] center: $4$}] {};
\end{scope}

\begin{scope}[xshift=5cm,yshift=0cm,start chain=going below,node distance=8mm]
\node[ssnode, on chain] (61) [label={[text=white] center: $1$}] {};
    \node[ssnode, on chain] (62) [label={[text=white] center: $1$}]
    {};
      \node[ssnode, on chain] (63) [label={[text=white] center: $1$}]
      {};
        \node[ssnode, on chain] (64) [label={[text=white] center: $1$}]
        {};
          \node[ssnode, on chain] (65) [label={[text=white] center:
            $3$}] {};
            \node[ssnode, on chain] (66) [label={[text=white] center: $3,4$}] {};
\end{scope}

\foreach \i in {1,2,...,6}
\draw (\i 1) -- (\i 2);

\foreach \i in {1,2,...,6}
\draw (\i 2) -- (\i 3);

\foreach \i in {1,2,...,6}
\draw (\i 3) -- (\i 4);

\foreach \i in {1,2,...,6}
\draw (\i 4) -- (\i 5);

\foreach \i in {1,2,...,6}
\draw (\i 5) -- (\i 6);

\foreach \i in {1,2,...,6}
\draw (1\i) -- (2\i);

\foreach \i in {1,2,...,6}
\draw (2\i) -- (3\i);

\foreach \i in {1,2,...,6}
\draw (3\i) -- (4\i);

\foreach \i in {1,2,...,6}
\draw (4\i) -- (5\i);

\foreach \i in {1,2,...,6}
\draw (5\i) -- (6\i);

\end{tikzpicture}
}
\captionsetup{width=0.5\textwidth}
\caption{Construction from Lemma~\ref{lem:gridcon} for the case
  $k=4$. The node in the $i$th row and $j$th column is labelled by the
  $\{a \, \mid \, u_a \in M(v_{i,j})\}$.
   } \label{fig: gridcon}
\end{figure}
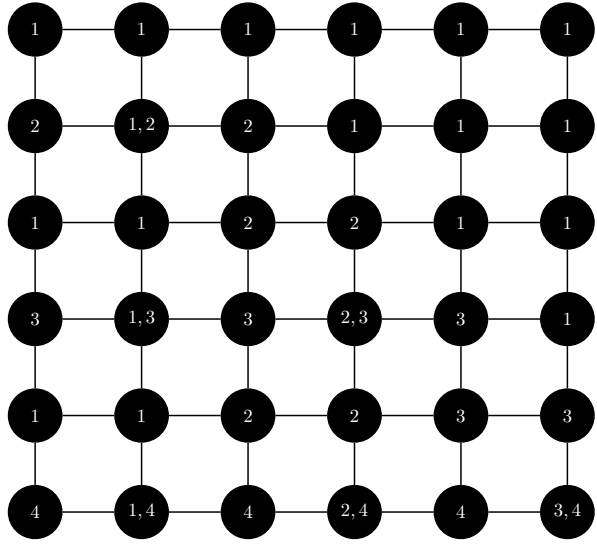

\begin{proof}[Proof of Lemma~\ref{lem:gridcon}]
Define 
\begin{align*}
X :=& V(\graphK_k) = \{u_i \, \mid \, i \in [k]\},\\ 
Y :=& V(\graphG_{2k-2})=\{v_{i,j} \, \mid \, i,j \in [2k-2]\},
\end{align*}
where $v_{i,j}$ is the vertex in the $i$th row and $j$th column of the grid and a map $M \colon Y \to 2^{X}$ as follows:
  \begin{enumerate}
  \item if $j-1>i$, $M(v_{i,j}) = \{u_1\}$,
  \item otherwise if $i \geq j-1$ then:
    \begin{enumerate}
  \item if $i$ and $j$ are both odd, $M(v_{i,j}) = \{u_{(j+1)/2}\} $,
    \item if $i$ is odd and is $j$ even, $M(v_{i,j}) = \{u_{j/2}\} $,
    \item if $i$ is even and $j$ is odd, $M(v_{i,j}) = \{ u_{(i+2)/2} \} $,
    \item if $i$ and $j$ are both even,  $M(v_{i,j}) = \{u_{(i+2)/2}, u_{j/2}\}$.
      \end{enumerate}
    \end{enumerate}

See Figure~\ref{fig: gridcon} for the case $k=4$. We claim this map is almost minor. Condition (1) in the definition is trivial and condition (4) follows directly from the definition of $M$. Moreover the following facts immediately imply that condition (2) holds:
\begin{itemize}
\item $M(v_{2i-2,1}) = \{u_i\}$, for every $i\in [k]\setminus\{1\}$,
\item $M(v_{1,1})=\{u_1\}$,
\item for any $i>j$, $M(v_{2i-2,2j-1}) = \{u_i\}$, $M(v_{2i-3,2j-1}) =\{u_j\}$ and,
\item $(v_{2i-2,2j-1}, v_{2i-3,2j-1}) \in E(\graphG_k)$.
\end{itemize} 

For (3) fix some $i\in [k]\setminus \{1\}$. Note that $u_i \in M(v_{a,b})$ if and
only if $(a,b)$ lies in one of the following
three categories.
\begin{enumerate}
\item $a$ odd and $b=2i-1$ with $a \ge 2i-1$.
\item $b=2i$ and $a \ge 2i-1$.
\item $a=2i-2$ and $2i-1 \ge b$.
\end{enumerate}

Let $T^i_{\alpha} =\{(v_{a,b}) \in Y \, \mid \, 
u_i \in M(v_{a,b}) \text{ and } (a,b) \text{ is in category } \alpha\}.$ Then $M^{-1}(\{u_i\}) = \bigcup_{i=1}^3 T_{\alpha}^{i}$. Clearly both $T_2$ and $T_3$ are 
connected. Furthermore since $v_{(a,2i-1)}$ is adjacent to
$v_{(a,2i)}$ for each odd $a$, $T_1 \cup T_2$ is also connected. Therefore, to prove $M^{-1}(\{u_i\})$ is connected it suffices
to show that there is a path from a vertex in $T_2$ to a
vertex in $T_3.$ $v_{2i-1,2i}, v_{2i-1,2i-1}, v_{2i-2,2i-1}$ is
such a path. Since the extra $(a,b)$ with $u_1 \in M(v_{a,b})$ are those where
$b-1>a$ and as $M(v_{1,b}) = \{u_1\}$ for all $b$, we also get connectivity in
this case. The result follows.
\end{proof} %
}

\begin{proof}[Proof of Theorem~\ref{thm:maintheorem}]

Let $\classA$ have unbounded treewidth.
Then for every $k$ there exists $\strucB_k \in \classA$ of treewidth at least $w(k)$. Then the Gaifman graph of $\strucB_k$, $\graphG_{\strucB_k}$ also has treewidth at least $w(k)$. By 
Theorem~\ref{thm: EG}, $\graphG_{\strucB_k}$ has $\graphG_k$ as a minor. Since by Lemma~\ref{lem:gridcon}, $\graphK_{(k+2/2)}$ is an almost minor of $\graphG_k$ we have:
\begin{align*}
\dsize_{\strucB_k}(m)\;&\underset{\text{\footnotesize (Lemma~\ref{lem:id-reduct})}}{=} 
\Omega\left(\dsize_{\strucB_k^{\text{id}}, \classC}(m)\right) \\ &\underset{\text{\footnotesize (Lemma~\ref{lem:Gaifman-reduct})}}{=}  
\Omega\left((\dsize_{\graphG_{\strucB_k}^{\text{id}}, \classH}(m))^{2/\ar(\strucB_k)}
\right) \\
 &\underset{\text{\footnotesize (Lemma~\ref{lem:minor-reduction})}}{=} 
\Omega\left((\dsize_{\graphG_k^{\text{id}},\classH'}(m))^{1/\ar(\strucB_k)}\right)  
\\ &\underset{\text{\footnotesize (Lemma~\ref{lem:almost-minor-red})}}{=} 
\Omega\left((\dsize_{\graphK_{(k+2)/2}^{\text{id}},\classD}(m))^{1/2\ar(\strucB_k)}\right) \\
&\underset{\text{\footnotesize (Lemma~\ref{lem:reductgraph})}}{=} 
\Omega\left(((\dsize_{\graphK_{(k+2)/2}}(m))^{1/2\ar(\strucB_k)}\right)  
\\ &\underset{\text{\footnotesize (Theorem~\ref{thm: d-clique})}}{=} 
 \Omega\left(m^{(k+2)/4r}/\log^{(3k+2)/2r}(m)\right) 
\end{align*}

Where $\classC$ is the class of $\sigma \cup \sigma_{B_k}$ structures $\strucC$ such that 
$\{P_a^\strucC \, \mid \, a \in B_k\}$ is a partition of the universe, $\classH$ the 
class of $V(\graphG_{\strucB_k}^{\text{id}})$-partitioned graphs, $\classH'$ the class of 
$V(\graphG_k^{\text{id}})$-partitioned graphs and $\classD$ the class of 
$V(\graphK_{(k+2)/2}^{\text{id}})$-partitioned graphs.
From the above we can conclude that (3) implies (1) in the statement of the theorem. Moreover, as discussed in Section~\ref{s: FR}, (1) implies (2) follows from \cite{olteanu2015} and (2) implies (3) trivially.
\end{proof}

\section{Conclusion}

Our main result characterises those bounded-arity classes of
structures $\classA$ where the set of homomorphisms from
$\strucA\in\classA$ to $\strucB$  can be succinctly represented. More
precisely, the known upper bound of
$O(\|\strucA\|^2\|\strucB\|^{\operatorname{tw}(\strucA)+1})$ is matched
by a corresponding lower bound of
$\Omega(\|\strucB\|^{\operatorname{tw}(\strucA)^\varepsilon})$, where
$\operatorname{tw}(\strucA)$ is the tree-width of $\strucA$ and
$\varepsilon>0$ is a constant depending on the excluded grid theorem and
the arity of the signature. A future task would be to further close the gap
between upper and lower bounds.

Another open question is to understand the representation complexity
for all classes of structures $\classA$ (of unbounded arity). As
mentioned in Section~\ref{s: FR}, a polynomial $O(\|\strucA\|^2\cdot\|\strucB\|^{\operatorname{fhtw}(\strucA)})$ upper bound was shown where $\operatorname{fhtw}(\strucA)$ is the fractional hypertreewidth of $\strucA$
\cite{olteanu2015} and one might wonder whether this is tight. At
least this is not the case in a parametrised setting, where a $f(\|\strucA\|)\|\strucB\|^w$ sized representation for some (not necessarily
polynomial-time) computable $f$, is considered \emph{tractable}. It is known that for structures $\strucA$ of \emph{bounded submodular
  width} $\operatorname{subw}(\strucA)$ the homomorphism problem can be decomposed into a (not
necessarily disjoint) union of
$f(\|\strucA\|)$ instances of bounded fractional hypertreewidth
\cite{marx2013,berkholz2020}, leading to a d-representation of size
$f(\|\strucA\|)\|\strucB\|^{\operatorname{subw}(\strucA)}$, \confORfull{see Appendix A in the full version of this paper for details}{see Appendix~\ref{sec: extendedUB}}{see Appendix~\ref{sec: extendedUB}}. Note that submodular width can be strictly
smaller than fractional hypertreewidth \cite{marx2011tractable}.
For a more concrete example in this direction, the
fractional hypertreewidth of $\graphC_4$ is 2, but one
can show that $\Hom(\graphC_4,\graphH)$ has deterministic
d-representations of size $O(\|\graphH\|^{3/2})$---almost matching the
$O(\|\graphH\|^{3/2}/\log^7(\|\graphH\|))$ lower bound in Example~\ref{ex: 3to4}.
Note that while submodular width characterises the FPT-fragment of
deciding the existence of homomorphisms on structures of unbounded
arity \cite{marx2013}, a tight characterisation for the
parameterised counting problem is, despite some recent progress \cite{sharpsub},  still
missing. In particular, it is not clear whether bounded submodular
width implies tractable counting. We may face similar difficulties when
studying the complexity of \emph{deterministic} d-representations that allow
efficient counting. 

In the course of proving our main result we have developed tools and
techniques for proving lower bounds on the size of d-representations,
in particular using our $k$-clique lower bound as a starting point,
defining an appropriate notion of reduction and showing that one can
always get such a reduction if the ``almost minor'' relation
holds. Whilst the proof of the clique lower bound in Section~\ref{s:
  clique} exploits the specific nature of d-representations, we
observe that much of the content of Section~\ref{s: dichotomy} can
easily be used for other forms of representations.
Since we now have understood the limitations of unrestricted
d-representations, it would be good to know whether there are even
more succinct representation formats that still allow efficient
enumeration.

 \vspace{0.5cm}
\textbf{Acknowledgements} Funded by the Deutsche Forschungsgemeinschaft (DFG, German Research Foundation) - project numbers 385256563;
414325841.

\bibliography{literature}

\clearpage
\appendix
\section{Extended Upper Bound Discussion} \label{sec: extendedUB}

In this section we expand the discussion from Section~\ref{s: FR} on upper bounds. After we review the bounds from \cite{olteanu2015} in more detail we highlight a new upper bound on representation size which, whilst it is directly implied by the results in \cite{berkholz2020}, has, to the best of our knowledge, not been observed elsewhere.

We begin by sketching how, given a structure $\strucA$ along with a tree decomposition of 
width $w-1$, we can generate d-reps of $\Hom(\strucA,\strucB)$ of size $O(\|\strucB\|^w)$. 
From the tree-decomposition we can infer information about how different elements of $
\mathcal{A}$ constrain the images of one another under a homomorphism. We may for example, 
learn that two sets of elements $X,Y$ are ``independent'', in the sense that the values a 
homomorphism takes on $X$ does not constrain the values that homomorphism takes on $Y$ 
and vice-versa. Therefore, if we can compute a d-rep $C_1$ of $\Hom(\mathcal{A}|_{X},
\mathcal{B})$ and $C_2$ of $\Hom(\mathcal{A}|_{Y},\mathcal{B})$, then we can construct a d-rep of $\Hom(\mathcal{A}|_{X \cup Y},\mathcal{B})$ by creating a new $\times$ gate and 
connecting the sources of $C_1$ and $C_2$ to this new gate. More generally we can infer a 
`factorisation structure', from which we can generate a d-rep of $\Hom(\mathcal{A},
\mathcal{B})$. To give a bit more detail, from a tree decomposition we can create a tree 
called a d-tree where each node is associated with an element $a \in \mathcal{A}$ and a 
set $\key(a)$ of elements of $\mathcal{A}$, which roughly correspond to all those 
elements lying above $a$ in the tree on which the image of $a$ depends on. Moreover it 
can be arranged that $\key(a)$ is a subset of a bag of our tree decomposition. This is 
what constrains the size of the produced d-rep, as well as the time needed to produce it.

If there is an ancestor of $a$ which is not in $\key(a)$ this process is not guaranteed to produce a f-rep. Therefore the class of d-reps considered in \cite{olteanu2015} is much more powerful than the class of f-reps considered. For example if $\mathcal{A}$ is a complete binary tree of height $h$, then for every structure $\mathcal{B}$, a d-rep of $\Hom(\mathcal{A}, \mathcal{B})$ can always be produced with size $O(2^h\|\mathcal{B}\|)$, since the treewidth of a tree is $1$. However there are arbitrarily large structures $\mathcal{B}$ such that any \emph{f-rep over a f-tree} of $\Hom(\mathcal{A}, \mathcal{B})$ has size $\Omega(\|\mathcal{B}\|^h)$. 

Indeed the results of \cite{olteanu2015} go further and in particular they provide the following tight upper bound on representation size, if we limit ourselves to the special class of \emph{d-reps over d-trees}.

\begin{theorem} \cite{olteanu2015} \label{thm: fhw}
Given a $\sigma$-structure $\mathcal{A}$ and a fractional hypertree
decomposition of $\mathcal{A}$ of width $w$, for any
$\sigma$-structure $\mathcal{B}$ we can compute a d-rep of
$\Hom(\mathcal{A}, \mathcal{B})$ of size
$O(\|\mathcal{A}\|^2\|\mathcal{B}\|^{w})$ in 
  time $O(poly(\|\mathcal{A}\|)\|\mathcal{B}\|^{w}
  \log\|\mathcal{B}\|).$

  Moreover there exists
  arbitrary large $\sigma$-structures $\mathcal{B}$ such that any d-rep over a
  d-tree of $\Hom(\mathcal{A}, \mathcal{B})$ has size
  $\Omega(\|\mathcal{B}\|^{w(\mathcal{A})})$, where $w(\mathcal{A})$ is the
    fractional hypertree width of $\mathcal{A}$. 
  \end{theorem}

 Given a class of structures $\classA$, Theorem~\ref{thm: fhw}
 tells us that the following are equivalent:
 \begin{enumerate}
 \item $\classA$ has bounded fractional hypertree width.
 \item There is a constant $c$, such that for every fixed $\sigma$-structure
   $\strucA \in \classA$ and any $\sigma$-structure
   $\strucB$, there is a d-rep over a d-tree of $\Hom(\strucA,\strucB)$ of size
   $O(\strucB^c)$. 
 \item For any fixed $\sigma$-structure
   $\strucA \in \classA$ there is a polynomial time algorithm, which for any 
   any $\sigma$-structure
   $\strucB$, constructs a d-rep over a d-tree of $\Hom(\strucA,\strucB)$.
 \end{enumerate}

 A priori, it could be that this equivalence breaks down when we consider general d-reps, not only d-reps over d-trees, but our main theorem implies that, on classes
 of bounded arity, this is not the case. This follows since a class of
 bounded arity has bounded fractional hypertree width if and only if
 it has bounded treewidth. In some sense therefore the approach of
 Olteanu and Z{\'a}vodn{\'y}, is the best we can do on classes of
 bounded arity. 

 What about on classes of unbounded arity? Here we lose the
 equivalence between bounded fractional hypertree width and
 bounded treewidth, meaning we cannot apply Theorem~\ref{thm:maintheorem} and so the situation is not as clear
 cut; we can use a width measure, introduced by Marx
 \cite{marx2013}, called submodular width to get smaller d-reps which
 allow for constant delay enumeration but which are not necessarily
 deterministic. The idea behind this approach is
 to use different tree decompositions on different parts of our 
 structure. In particular we use the results obtained by
 Berkholz and Schweikardt \cite{berkholz2020} in the context of
 enumerating the results of conjunctive queries and observe that they
 can easily be adapted. 
 
 \begin{theorem} \label{thm: subw}
   Let $\mathcal{A}$ and $\mathcal{B}$ be $\sigma$-structures, such
   that $\mathcal{A}$ has submodular width
   $\le w$. Then for all $\delta >0$ we can construct a $C \in \dset(\mathcal{A}, \mathcal{B})$ of size $O(g(\|\mathcal{A}\|) \cdot
   \|\mathcal{B}\|^{(1+\delta)w})$ in time $O(h(\|\mathcal{A}\|) \cdot \log(\|\mathcal{B}\|)
 \cdot  \|\mathcal{B}\|^{(1+\delta)w})$, where $g, h$ are computable
 functions. 

   \end{theorem}

   \begin{proof}[Proof (sketch)] 
 Fix some $\delta
 >0$, and suppose we have a $\sigma$-structure
 $\mathcal{A}$ with submodular width $\le w$. Using the results from \cite{berkholz2020} construct $\ell$ pairs of structures
 $(\strucA_1,\mathcal{B}_1), \dots, (\strucA_{\ell},\mathcal{B}_{\ell})$ such that
$\Hom(\mathcal{A},\mathcal{B}) = \cup_{i=1}^{\ell}
\Hom(\mathcal{A}_i,\mathcal{B}_i)$, where $\ell$ is bounded by a
function of $\|\mathcal{A}\|$. By the bound on the submodular width, for each
 pair we have an associated tree decomposition $(T_i,
 \beta_i)$ of $\mathcal{A}$ which has certain `nice' properties, crucially it is also a tree decomposition of $\mathcal{A}_i$. We
 can then apply a construction from \cite{olteanu2015} (Proposition
 9.3.) to turn each tree decomposition into a d-tree,
 $\mathcal{T}_i$. By the `niceness' of the original tree decomposition
 we get that the d-rep of  $\Hom(\mathcal{A}_i,\mathcal{B}_i)$ over $\mathcal{T}_i$ has size
 at most $O(\|\mathcal{A}\|^2 \cdot \|\mathcal{B}\|^{(1+\delta)w})$. We can
 then join the sinks of each of these d-reps to a new $\cup$-gate and
 we end up with a d-rep of $\Hom(\mathcal{A},\mathcal{B})$. The time
 bound follows from the time bounds given in \cite{berkholz2020} and \cite{olteanu2015}.
\end{proof}

The above theorem really does give us smaller representations. To see this observe that the same lower bound on $\dsize_{\graphC_4}$ from Example~\ref{ex: 3to4} also applies to $\dsize_{\graphC_k}$, for any $k>2$, since $\graphK_3$ is a minor of $\graphC_k$. Since the submodular
 width of $\graphC_k$ is $3/2$ for all $k>2$, we get that this lower bound is tight up to a logarithmic
 factor. On the other hand the fractional hypertree width of $\graphC_k$ is 2, meaning that within the class of d-reps over d-trees we have a lower bound of $m^2$. There is also an improvement in terms of classes of structures, since there are classes of bounded submodular width but unbounded fractional hypertree width (see
 \cite{marx2011tractable}, note that  bounded adaptive width and
 bounded submodular width are equivalent properties \cite{marx2013}).

 Since the $\ell$ in the proof of Theorem~\ref{thm: subw} only depends on
$\mathcal{A}$, by a
result from Durand and Strozecki \cite{durand2011}, the d-rep
constructed in the theorem can be used for constant-delay enumeration
of $\Hom(\mathcal{A}, \mathcal{B}).$  We can supplement this with a
partial converse based on the following conditional lower bound 
from \cite{berkholz2020}: assuming the exponential time hypothesis (ETH) for any
recursively enumerable class of self-join free structures $\classA$
there is an algorithm for Enum-CSP($\classA$, $\classall$) with
constant delay after FPT preprocessing if and only if $\classA$ has
bounded submodular width. Therefore, in particular, assuming ETH there
is a FPT algorithm which given a $\sigma$-structure $\strucA \in \classA$ and
a $\sigma$-structure $\mathcal{B}$, produces within FPT-time a d-rep which
allows us to enumerate $\Hom(\mathcal{A}, \mathcal{B})$ with delay
depending only on $\mathcal{A}$ if and only if
$\classA$ has bounded submodular width. 

However note that the d-reps produced by
Theorem~\ref{thm: subw} are not necessarily deterministic and cannot be used for counting. An interesting question is whether this separation between
counting and enumeration is intrinsic or down to our ignorance, see \cite{sharpsub} for work in this direction. Approaching this question from the perspective of this paper we would like to find either an algorithm that produces
deterministic d-reps of homomorphisms from classes of bounded
submodular width in FPT time or a proof that such an algorithm is
impossible.

\section{Proof of Lemma~\ref{lem:technical}}
\begin{proof}
To construct $C'$ first remove all input gates of the form $\{a \mapsto b\}$ for $a 
\not\in X$, then iteratively delete all gates which have no in-edges. At each stage we 
also delete any edges which originate from a deleted gate. This construction can be 
carried out via a single bottom up traversal of $C$ and therefore in time linear in $\|
C\|$. It is intuitively easy to see the construction has the claimed properties although 
the formal proof below is a little involved. 

Let $h \in \pi_{X}\Hom(\mathcal{A},\mathcal{B})$. Then there is some $h' \in S_{C}$ such 
that $\pi_{X} h' = h$. Let $g_1, \dots, g_{\alpha}=s$ be a path in $C$, such that $
\pi_{\dom(g_i)}h' \in S_{g_i}$ and $g_1$ is of the form $\{a \mapsto b\}$ for some $a \in 
X$. Such a path exists since $h \in C_g$, so we can construct this path by choosing an 
arbitrary $a \in X$, and working backwards starting at the source, such that if we have 
constructed the path back to gate $g$, we next choose some child $c$ of $g$, such that $
\pi_{\dom(c)}h' \in S_{c}$ and $a \in \dom(c)$. Note that all of the $g_i$ are also gates 
in $C'$. We prove by induction that $\pi_{\dom(g_i) \cap X} h' \in S_{g_i}^{C'}$, where 
$S_{g_i}^{C'}$ refers to the set of (partial) homomorphisms defined by Equation~\eqref{eq:1} for the circuit $C'$. We also write $\dom^{C'}(g)$ to denote the domain of gate $g$ in $C'$, and retain the notation $S_g$ and $\dom(g)$ to refer to the original circuit $C$. The base case 
is trivial from the definition of the path. Moreover note that for any gate $g$ in $C'$, 
$ \dom^{C'}(g) = \dom(g) \cap X$, this can be proved by a simple bottom up induction on 
$C'$. Suppose the claim is true for $g_i$ and that $\dom^{C'}(g_{i+1}) = \dom^{C'}(g) 
\cup \bigcup_{i=1}^t y_i$. Suppose $h'(y_i) = b_i$. Since $\pi_{\dom(g_{i+1})}h' \in 
S_{g_{i+1}}$, it follows there is an input gate of the form $\{y_i \mapsto b_i\}$ in 
$C_{g_{i+1}}$. Then if we consider the path in $C$ from this gate to $g_i$ it is clear 
that none of the gates or wires are deleted and so $\{y_i \mapsto b_i\}$ is a gate in 
$C'_{g_i}$. This implies that $\pi_{\dom(g_{i+1}) \cap X} h' = \pi_{\dom^{C'}(g_i)}h' \times \{y_1 
\mapsto b_1\} \times \dots \times \{y_t \mapsto b_t\} \in S^{C'}_{g_{i+1}}$. We therefore 
have that $h = \pi_X h' \in S_{g_\alpha} = S_{C'}$. Conversely one can show by bottom up 
induction on $C'$ that for every gate $g$ and every $h \in S^{C'}_g$, $h \in 
\pi_{\dom^{C'}(g)} \Hom(\mathcal{A}, \mathcal{B})$ which, combined with the observation 
that $\dom^{C'}(g) =\dom(g) \cap X$ completes the proof that $S_{C'}= \pi_{X}
\Hom(\mathcal{A},\mathcal{B})$. Finally in constructing $C'$ we only removed wires and 
gates from $C$ and so clearly $\|C'\| \le \|C\|$.

To construct $C''$ first remove all input gates of the form $\{x_i \mapsto b\}$ for $x_i \in X$, $b \not \in Y_i$. Then iteratively perform the following procedure:
\begin{itemize}
\item If $g$ is a $\cup$-gate and $g$ has no in-edges delete $g$
\item If $e=\{g', g\}$ is an edge in our circuit and $g'$ is deleted then delete $e$. If $g$ is a $\times$-gate also delete $g$. 
\end{itemize} 
As before this construction can be carried out via a single bottom up traversal of $C$ and therefore in time linear in $\|C\|$. Correctness follows from a very similar inductive argument to the above. 
\end{proof} %
 \label{ap: technical}

\section{Proof of Claim~\ref{claim: big}} \label{ap: d-clique}

\begin{proof}
We begin by proving the following claim which is the key to bounding the size of $C'$. 
\begin{claim} \label{claim: active}
Let $g$ be a definition in $\hat{C}$ such that there is a unique path from each parent of $g$ to $s$. Then $g$ has at most $3^{k-1}\log^{k-1}(n)$ active parents.
\end{claim}
\begin{proof}
First observe that we may assume WLOG that the domain of any definition in $C$ is a strict subset of $\{x_1,\ldots,x_k\}$, as otherwise we could replace $C$ with an equivalent but smaller d-rep where this property holds.\footnote{In more detail: note that if $\dom(g) = \{x_1, \dots, x_k\}$ then any out edge from $g$ goes to a $\cup$-gate. Let $(g,p)$ be such an edge. Consider the circuit where we replace this edge with an edge from $g$ to $s$. Clearly this is an equivalent circuit. By repeating this process for every such edge we get an equivalent circuit which is no bigger than our original and where there is no definition with domain $\{x_1,\ldots,x_k\}$.} Suppose that there are $\beta$ active parents of a gate $g$ in $\hat{C}$. If $h, h' \in S_g$ then for every mapping  $h^{\ast}$
with domain $\dom(C) \setminus \dom(g)$ it holds that $h
\cup h^{\ast} \in A_i$ if and only if $h' \cup h^{\ast} \in A_i$,
since the sets $A_i^v$ are only expanded at $\times$-gates. Therefore,
if $h \cup  h^{\ast} \in A_i \setminus \cup_{i \neq j}A_j $ then so is
$h' \cup h^{\ast}$ for all $h' \in S_g$.  It follows that $|\!\cup_{i}A_i| \ge \beta|S_g|$ and so 
$|\!\cup_{i} \pi_{\dom(C) \setminus \dom(g)} A_i | \ge
\beta$. Since these correspond to $\beta$ distinct $(k-|\!\dom(g)|)$-cliques in
$G$, there must be
$x_1\in \dom(C) \setminus \dom(g)$ such that $|\pi_{\{x_1\}}\big(\!\cup_{i} \pi_{\dom(C) \setminus \dom(g)} A_i\big) | \ge
\beta^{1/(k-|\!\dom(g)|)}\ge \beta^{1/(k-1)}$. 
Since $g$ is big there is also a $x_2\in\dom(g)$ such that
$|\pi_{\{x_2\}}S_g|>3\log(n)$. By the considerations above, the images of
both sets form a biclique  of size at least $\beta^{1/(k-1)}\times
3\log(n)$ in $\graphG$. Therefore, by Corollary~\ref{thm:
  G},  $\beta \le 3^{k-1}\log^{k-1}(n)$.
\end{proof}

We now show how this claim implies the size bound. Consider an edge $E$ going from $g$ to $p$ in $\hat{C}$. When our transversal reaches $p$, it has also visited every parent of $p$ so there is a unique path from every such parent to $s$. Consequently, we may apply Claim~\ref{claim: active} and so we make $\beta \le 3^{k-1}\log^{k-1}(n)$ copies of $p$; $E$ is copied the same number of times. Denote the copies of $p$ by $p_i, i \in [\beta]$ and the copy of $E$ going from $g$ to $p_i$ by $E_i$. When our transversal reaches $g$, each $E_i$, is either deleted or there is a unique copy of $g$ which has an edge to $p_i$ and we may regard this edge as being $E_i$. At later stages of the traversal the edges $E_i$ could be deleted but no new copies are made. Therefore the number of wires in $C'$ is at most $3^{k-1}\log^{k-1}(n)$ 
bigger than the number of wires in $\hat{C}$. Moreover note that the size of any f-rep is 2 times the number of wires plus 1. Therefore $\|C'\| 
\le 2 \cdot 3^{k-1}\log^{k-1}(n)\|\hat{C}\|$. It remains to show that $C'$ is indeed an equivalent 
circuit.
                                  
So suppose that up to the stage where our traversal visit gate $g$, with parents $p_1,\dots, 
p_{\alpha}$, the resulting circuit is equivalent to $\hat{C}$. Let $C(g)$ 
be the circuit formed by deleting $g$ and creating $\alpha$ new gates $g_1, \dots, 
g_{\alpha}$ such that the children of each $g_i$ are exactly the children of $g$ and 
$g_i$ has exactly one out-edge going to $p_i$. Clearly the resulting circuit is 
equivalent to $\hat{C}$. Now suppose $p=p_i$ is a parent of $g$ and that $p$ is not 
active. We form the circuit $C(g,p)$ from $C(g)$ by deleting the wire from $g$ to $p$ and then 
iteratively deleting all gates that have no incoming edges as well as all edges 
originating from such gates. Suppose $h \in A_j$ with $j \neq i$. Then $h\in S_{C(g,p)}$, 
since none of the wires or gates on the path
from $g$ to $s$, via $p_j$ have been altered. Therefore by construction
$S_{C(g,p)} \supseteq S_{C(g)} \setminus A_i$. Since $A_i \subseteq
\cup_{i \neq j} A_j$, $S_{C(g,p)}=S_{C(g)}$ and we can get a smaller
equivalent circuit by replacing $C(g)$ with $C(g,p)$. By repeating this process for every non-active parent of $g$, we get exactly the circuit that our process produces after it has visited $g$ and by the above such a circuit is equivalent to $\hat{C}$. 
\end{proof}

\end{document}